\ifx\XeTeXversion\undefined\pdfoutput=1\fi
\PassOptionsToPackage{hyphens}{url}   %
\documentclass[preprint,12pt]{elsarticle}
\biboptions{sort&compress}

\journal{Journal of Computational Physics}

\newcommand{\blind}{0}

\usepackage[T1]{fontenc}
\usepackage[utf8]{inputenc}
\usepackage{lmodern}
\usepackage{amsmath,amssymb,amsfonts,amsthm,mathtools}
\usepackage{graphicx}
\usepackage{booktabs}
\usepackage{array}
\usepackage{subcaption}
\usepackage{float}
\newlength{\figwidth}
\newlength{\fieldwidth}
\usepackage{caption}
\usepackage[shortlabels]{enumitem}
\usepackage{xcolor}
\usepackage{url}
\usepackage{pifont}      %
\usepackage{makecell}    %
\usepackage{tikz}
\usetikzlibrary{arrows.meta,positioning,calc,fit,backgrounds,shapes.geometric}
\usepackage{algorithm}
\usepackage{algpseudocode}
\usepackage[hidelinks]{hyperref}
\makeatletter
\def\ps@pprintTitle{\let\@oddhead\@empty\let\@evenhead\@empty
  \def\@oddfoot{\@myfooterfont Preprint submitted to \ifx\@journal\@empty Elsevier\else\@journal\fi\hfill\@date}%
  \let\@evenfoot\@oddfoot}
\makeatother
\usepackage{lineno}
\usepackage{natbib}   %

\algrenewcommand\algorithmicrequire{\textbf{Input:}}
\algrenewcommand\algorithmicensure{\textbf{Output:}}

\newtheorem{assumption}{Assumption}
\theoremstyle{remark}

\theoremstyle{plain}

\newtheorem{proposition}{Proposition}

\theoremstyle{definition}

\newcommand{\Gr}{\mathrm{Gr}}

\newcommand{\vect}[1]{\boldsymbol{#1}}

\begin{document}

\begin{frontmatter}

\if0\blind
\title{A Product-Grassmann Tucker Surrogate for Parametric Field Prediction}
\author[gatech]{Melika Baghi}
\ead{mbaghi3@gatech.edu}
\author[gatech]{Xiao Liu}
\ead{xliu976@gatech.edu}
\author[gatech]{Kamran Paynabar\corref{corresponding}}
\ead{kpaynabar3@gatech.edu}
\cortext[corresponding]{Corresponding author}
\address[gatech]{H. Milton Stewart School of Industrial and Systems Engineering, Georgia Institute of Technology, Atlanta, GA 30332, USA}
\fi

\if1\blind
\title{A Product-Grassmann Tucker Surrogate for Parametric Field Prediction}
\author{Anonymous authors}
\fi

\begin{abstract}
Many-query applications such as design optimization, uncertainty quantification and digital twins call for repeated full-field predictions from expensive partial differential equation (PDE) simulations. Grassmannian reduced-order models interpolate parameter-dependent subspaces but commonly vectorize the spatial field, and the ambient dimension then grows with the product of the grid sizes. Tensor-based models preserve the multiway structure but often use shared mode-wise bases or a fixed core across parameters. We propose a Tucker-based surrogate that keeps the spatial directions and time as separate modes. The factor subspaces are interpolated on a product Grassmann manifold, Procrustes alignment places the Tucker cores in one coordinate system, and a truncated singular value decomposition models their variation. Across five PDE benchmarks in two and three spatial dimensions, including public computational fluid dynamics data, the interpolation size is reduced by factors of $13$ to $482$ and the stored model by factors of $13$ to $464$ relative to a single-subspace Grassmannian surrogate with the same local interpolation. Cached queries run $3.3$ to $8.5$ times faster, the mean errors over ten resampled splits differ by at most $0.22$ percentage points ($0.36$ when every model choice is re-selected inside each split), and the prediction-error terms of the two methods differ by less than $0.1$ percentage points. The remaining accuracy differences are consistent with the larger representation error of the Tucker format. In a three-dimensional grid study the interpolation-size ratio rises to $2{,}260$ and the query speedup to $11.9$ at $64^3$.
\end{abstract}


\end{frontmatter}

\section{Introduction}
\label{sec:introduction}

Many-query applications such as design optimization, uncertainty quantification and digital twins evaluate a parameterized physical model again and again~\citep{deo2025physicsguided,musekamp2025al4pde}. Increasingly, they need the full spatial or spatiotemporal field, not a few scalar quantities~\citep{tian2025distortion,tian2025temperature,thiamchaiboonthawee2026flare}. A surrogate for this setting has to approximate an expensive simulator and, at the same time, represent and predict its high-dimensional output at low cost~\citep{giacomini2026surrogates}.

Let $\mu$ collect the material properties, boundary conditions or operating settings of the model. For a two-dimensional spatial domain observed over time, one simulation produces a third-order tensor
\begin{equation}
\label{eq:intro_state_tensor}
\mathcal{X}(\mu)\in\mathbb{R}^{X\times Y\times T},
\end{equation}
where $X$ and $Y$ index the spatial grid and $T$ the stored time steps. As $\mu$ varies, the dominant structures of the field and their amplitudes change together.

Projection-based reduced-order models commonly reshape the spatial field into a vector of length $XY$ and represent the solution in a low-dimensional subspace~\citep{mosquera2019idw,oulghelou2021bicitsgm,friderikos2022stability}. When that subspace varies with the parameter, Grassmannian interpolation treats it as a point on a manifold. Each local basis is represented by the subspace it spans, mapped to a tangent space, interpolated there and mapped back to the manifold. Once the spatial dimensions are reshaped into a single axis, these operations act in an ambient space of dimension $XY$, and the object that is stored and interpolated grows with the product of the spatial dimensions. This study measures the cost and accuracy of keeping the tensor modes separate, relative to the vectorized formulation, and examines the representation and prediction errors separately.

This work preserves the multiway structure of the field. Each training field is compressed by a Tucker decomposition, with one factor matrix for each of the $X$, $Y$ and $T$ modes. The column spaces of the three factors define a point on the product manifold
\begin{equation}
\label{eq:intro_product_manifold}
\mathcal{M}=\Gr(r_X,X)\times\Gr(r_Y,Y)\times\Gr(r_T,T),
\end{equation}
where $\Gr(r,n)$ is the set of $r$-dimensional subspaces of $\mathbb{R}^n$.

Preserving the tensor structure raises one further difficulty. A Tucker decomposition is not unique: equivalent factor bases differ by orthogonal rotations that the core absorbs, and cores computed at different parameters therefore lie in different coordinate systems~\citep{DeLathauwer2000a,Kolda2009}. The training factors are aligned to a common reference by orthogonal Procrustes rotations~\citep{Schoenemann1966,Higham1986}, their cores are rotated to match, and a truncated singular value decomposition (SVD) of the aligned core residuals gives a low-dimensional model of the core variation, the low-rank aligned-core model. The factor subspaces and the core coordinates are then predicted with the same local parameter-space weights. The present implementation uses one reference tangent space and is intended for parameter ranges in which the mode-wise subspaces remain close enough to that reference.

The mode-wise representation also shrinks the Grassmann interpolation problem (Figure~\ref{fig:concept}). At comparable ranks the interpolated object scales with the individual mode sizes in place of their spatial product. The mode-wise and the single-subspace surrogates still reconstruct the same full field, and that stage bounds the speedup of the whole query.

\begin{figure}[t]
\centering
\resizebox{\linewidth}{!}{%
\begin{tikzpicture}[font=\normalsize, >=Latex,
  box/.style ={draw, rounded corners=2pt, align=center, minimum height=1.05cm, inner sep=5pt, fill=blue!7},
  good/.style={draw, rounded corners=2pt, align=center, minimum height=1.05cm, inner sep=5pt, fill=green!10},
  warn/.style={draw, rounded corners=2pt, align=center, minimum height=1.05cm, inner sep=5pt, fill=red!10},
  recb/.style={draw, rounded corners=2pt, align=center, minimum height=1.05cm, inner sep=5pt, fill=orange!14},
  lab/.style ={font=\normalsize\bfseries},
  ar/.style  ={->, thick}]
\fill[red!3,rounded corners]   (-0.7, 1.05) rectangle (17.4, 3.25);
\fill[green!3,rounded corners] (-0.7,-3.55) rectangle (17.4, 0.55);
\node[lab,anchor=west] at (-0.6,3.65) {(a) Single-subspace Grassmannian surrogate};
\node[box]  (fa) at (0.8, 2.15) {$\mathcal{X}(\mu)$};
\node[warn] (ra) at (4.0, 2.15) {Reshape\\$X\times Y\rightarrow XY$};
\node[warn] (ga) at (10.1,2.15) {One subspace\\$\mathrm{Gr}(r_s,XY)$};
\node[recb] (xa) at (15.7,2.15) {Reconstruction};
\draw[ar] (fa)--(ra); \draw[ar] (ra)--(ga); \draw[ar] (ga)--(xa);
\node[lab,anchor=west] at (-0.6,0.95) {(b) Product-Grassmann Tucker surrogate};
\node[box]  (fb) at (0.8,-1.5) {$\mathcal{X}(\mu)$};
\node[good] (tb) at (4.0,-1.5) {Tucker\\decomposition};
\node[good] (pb) at (10.1,-0.35) {Factor subspaces\\$\mathrm{Gr}(r_X,X)\times\mathrm{Gr}(r_Y,Y)\times\mathrm{Gr}(r_T,T)$};
\node[good] (cb) at (10.1,-2.65) {Aligned core model};
\node[font=\small] at (10.1,-1.5) {same weights $w(\mu)$};
\node[recb] (xb) at (15.7,-1.5) {Reconstruction};
\draw[ar] (fb)--(tb);
\draw[ar] (tb.east) -- ++(0.55,0) |- (pb.west);
\draw[ar] (tb.east) -- ++(0.55,0) |- (cb.west);
\draw[ar] (pb.east) -| (xb.north);
\draw[ar] (cb.east) -| (xb.south);
\end{tikzpicture}%
}
\caption{The single-subspace Grassmannian surrogate (a) and the product-Grassmann Tucker surrogate (b).}
\label{fig:concept}
\end{figure}

The main contributions of this paper are the following.
\begin{itemize}[leftmargin=1.4em,itemsep=2pt,topsep=2pt]
\item A product-Grassmann Tucker surrogate predicting parameter-dependent mode-wise subspaces.
\item A reference-aligned low-rank model of parameter-dependent Tucker cores with a compatible reconstruction.
\item Cost and error analyses validated on five partial differential equation (PDE) benchmarks, including public computational fluid dynamics (CFD) data and a three-dimensional scaling study.
\end{itemize}

\section{Related Work}
\label{sec:related_work}

\subsection{Parametric reduced-order models and Grassmann interpolation}

Classical surrogates such as radial basis functions and Gaussian processes target scalar or moderate-dimensional outputs~\citep{Buhmann2003,Rasmussen2006}, while each parameter here yields a high-dimensional field. Projection-based reduced-order models represent such fields in low-dimensional spaces built from simulation snapshots~\citep{Benner2015,Quarteroni2016,Hesthaven2016,Benner2017Book,Sirovich1987,Berkooz1993,Willcox2002,Rowley2004}. A single global basis grows inefficient when the dominant solution structure changes across the parameter domain, which motivates local and parameter-dependent reduced spaces~\citep{Amsallem2012LocalROM}. Grassmannian methods represent these spaces as points on a Grassmann manifold~\citep{Edelman1998,Absil2008,Kobayashi1963} and interpolate them through tangent coordinates~\citep{Amsallem2008,Son2013}, an approach also applied to reduced operator matrices on other matrix manifolds~\citep{amsallem2009interpolating,Amsallem2011}. Later work develops the geometry, stability and extensions of that procedure~\citep{Bendokat2020,zimmermann2021manifold,friderikos2022stability,Zhang2022GPS,ciaramella2026grassmann,sreekumar2026roadmap}, reads parametric sensitivities as tangent vectors~\citep{sato2025parametric} and raises the interpolation order~\citep{mosquera2021neville,niu2026stable,jensen2025maxvol}. Other extensions use alternative manifold interpolants~\citep{oulghelou2021barycentric,elomari2025stiefel,mosquera2024kriging}, learn the map from parameter to subspace~\citep{Giovanis2020GrassmannGP,liu2023regression,liu2026projected}, use Grassmann distances for adaptive sampling~\citep{giovanis2018uncertainty} or carry the interpolation over to dynamic mode decomposition modes~\citep{du2026dmd}.

Many Grassmannian reduced-order models vectorize the spatial state first~\citep{Amsallem2008,Son2013,mosquera2019idw,oulghelou2021bicitsgm,friderikos2022stability}. For a two-dimensional field the spatial basis then lives on $\Gr(r_s,XY)$, and the ambient dimension of the manifold grows with the product of the grid dimensions.

Operator inference learns the reduced operators from snapshot data~\citep{peherstorfer2016opinf}, and neural operators, transformer solvers and hybrid neural reduced-order models learn the parameter-to-field map itself~\citep{li2021fourier,lu2021deeponet,li2023scalable,wu2024transolver,luo2025transolverpp,wen2025gaot,fresca2022poddlrom}.

\subsection{Tensor-structured reduced-order modeling}

Tensor decompositions preserve the multiway organization that reshaping the spatial dimensions into one axis removes. The Tucker decomposition, computed through the higher-order singular value decomposition (HOSVD), expresses a tensor as a compact core and one factor matrix per mode~\citep{DeLathauwer2000a,DeLathauwer2000b,Kolda2009,Vannieuwenhoven2012}, and tensor-train and hierarchical formats extend it to larger or higher-order data~\citep{Oseledets2011,Hackbusch2009,Grasedyck2010,Nouy2017}. With the parameters as additional tensor modes, low-rank formats solve parameterized linear systems and approximate parameter-dependent outputs~\citep{kressner2011lowrank,ballani2015hierarchical}. Several reduced-order models treat the parameter as one more mode of the stacked training data, build space-time trial subspaces from tensor decompositions of the snapshots, or regress the coefficients of common mode-wise bases on the inputs~\citep{Mamonov2022,Mamonov2024,Olshanskii2024,islam2026tensorial,budzinskiy2025cross,choi2019spacetime,gopalan2022hosvd,karasozen2022nonintrusive,barragan2025motif}. These models preserve the tensor structure, and many of them share their physical-mode bases across the parameter domain.

With the parameter as a mode, the parameter dependence can also be carried by an interpolated parameter-mode factor. Sasagawa and Tanaka interpolate its entries with Akima splines while the core and the spatial and temporal factors are fixed~\citep{sasagawa2023tucker}, and a hybrid variant interpolates it by Fourier series and corrects the interpolation artifacts with a neural network~\citep{shaa2026hybrid}. Product-Grassmann representations describe the factor subspaces of a single Tucker decomposition in Newton, trust-region and preconditioned optimization methods~\citep{elden2009newton,ishteva2011best,desterck2016tucker} and in gesture recognition~\citep{lui2012product}, with no parametric prediction.

Closest to the present construction are space-time methods in which both the spatial and the temporal basis vary with the parameter. Oulghelou and Allery interpolate the two bases on Grassmann manifolds and align the interpolated modes with two orthogonal calibration matrices~\citep{oulghelou2021bicitsgm}. Friderikos et al.\ interpolate the spatial and the temporal basis separately on Grassmann manifolds~\citep{friderikos2020grassmann}. Their later method follows the two bases as curves on Stiefel manifolds, interpolates a small square matrix built from both bases and the singular values, and makes the singular vectors unique by an oriented singular value decomposition~\citep{friderikos2021spacetime}. In these methods one basis spans all spatial degrees of freedom. The present surrogate extends this construction to three and four modes, with a factor subspace for each spatial axis and for time on a product Grassmann manifold. Its full Tucker core is aligned to a reference mean and modeled by a truncated singular value decomposition of its residuals, and the interpolated object grows with the sum of the mode sizes. Tucker is the format used because each physical mode has an orthonormal factor subspace that is interpolated directly on a Grassmann manifold. By contrast, the subspaces of a tensor train mix a physical axis with a neighboring rank, and the factors of a canonical polyadic (CP) decomposition are not orthonormal bases~\citep{Kolda2009,Oseledets2011}.

Tensor regression uses Tucker or CP structure on coefficient tensors for statistical efficiency~\citep{Zhou2013,Li2018TuckerRegression,fang2017image}. Proper generalized decomposition and its hierarchical deep-learning versions build a separated representation from the governing equations before any snapshot exists~\citep{chinesta2011pgd,zhang2022hidennpgd,lu2023chidenn}, and interpolating neural networks with tensor decomposition extend that construction~\citep{park2025inn,guo2025inntd,guo2025taps}.

\subsection{Basis alignment}

Tucker factors are not unique, since right-multiplying a factor by an orthogonal matrix and applying the inverse rotation to the corresponding core mode leaves the reconstructed tensor unchanged~\citep{DeLathauwer2000a,Kolda2009,koch2010dynamical,uschmajew2013geometry}. Across parameters the ambiguity has to be removed before the cores are compared. Orthogonal Procrustes alignment is the standard way to place equivalent bases in one orientation~\citep{Schoenemann1966,Higham1986}, and alignment before interpolation has been applied to vectorized reduced bases through a consistent mode ordering~\citep{goutaudier2023orbi} and through the removal of sign and rotation ambiguities~\citep{du2026mrpwi}. A recent generative model aligns the Tucker factors of spatiotemporal data to medoid anchors by the same Procrustes rotations~\citep{lyu2026diffats}.

\subsection{Limits of existing approaches}

Three limits remain in the approaches reviewed above. Grassmannian surrogates that vectorize the field interpolate one subspace whose ambient dimension is the product of the grid sizes. Tensor-structured models keep the modes, but they typically share the physical-mode bases across the parameter domain and carry the parameter dependence in the core, in reduced coefficients or in a parameter-mode factor. The space-time methods let both bases vary with the parameter, yet one basis still spans all spatial degrees of freedom. In the present surrogate every Tucker mode has its own parameter-dependent subspace on a product Grassmann manifold, and the Tucker core, aligned to a common reference, varies with the parameter as well.

\section{Methodology}
\label{sec:method}

The surrogate predicts a full spatiotemporal field at a new parameter from a set of training simulations. Each training field is first compressed by a Tucker decomposition. The mode-wise factor subspaces are then represented and interpolated on a product Grassmann manifold, and the Tucker cores are aligned to a common coordinate frame and compressed into a low-rank aligned-core model. At a query, the predicted factors and core are combined into the field.

\subsection{Problem setup and Tucker representation}
\label{ssec:tucker}

Let $\mathcal{P}\subset\mathbb{R}^d$ be the parameter domain. For every $\mu\in\mathcal{P}$ the high-fidelity solver returns the spatiotemporal field $\mathcal{X}(\mu)\in\mathbb{R}^{X\times Y\times T}$ of \eqref{eq:intro_state_tensor}, where $X$ and $Y$ are the two spatial grid dimensions and $T$ the number of stored time steps. The training set holds $N$ simulations, $\{(\mu_i,\mathcal{X}(\mu_i))\}_{i=1}^N$. Every training field keeps its three-mode structure through a truncated Tucker decomposition,
\begin{equation}
\label{eq:tucker_approx}
\mathcal{X}(\mu_i)\approx
\mathcal{G}_i \times_1 U_{X,i} \times_2 U_{Y,i} \times_3 U_{T,i},
\end{equation}
where $U_{X,i}\in\mathbb{R}^{X\times r_X}$, $U_{Y,i}\in\mathbb{R}^{Y\times r_Y}$ and $U_{T,i}\in\mathbb{R}^{T\times r_T}$ have orthonormal columns, $\mathcal{G}_i\in\mathbb{R}^{r_X\times r_Y\times r_T}$ is the Tucker core, with $r_X<X$, $r_Y<Y$ and $r_T\le T$, and $\times_n$ is the mode-$n$ product~\citep{Kolda2009}. The ranks set both the compression and the size of the online problem.

The Tucker representation separates two forms of parameter dependence: the factor matrices describe how the dominant spatial and temporal subspaces change with $\mu$, and the core describes how the coordinates along those subspaces interact.

\subsection{Product-Grassmann representation of the factor subspaces}
\label{ssec:product}

An orthogonal rotation of the columns of a factor matrix in \eqref{eq:tucker_approx} leaves the spanned subspace unchanged. The prediction target is therefore the column space of each factor,
\[
[U_{X,i}]\in\Gr(r_X,X),\qquad
[U_{Y,i}]\in\Gr(r_Y,Y),\qquad
[U_{T,i}]\in\Gr(r_T,T),
\]
where $\Gr(r,n)$ is the Grassmann manifold of $r$-dimensional subspaces of $\mathbb{R}^n$. Together the three mode-wise subspaces define the product point
\begin{equation}
\label{eq:product_point}
\mathbf{U}_i
=
\bigl([U_{X,i}],[U_{Y,i}],[U_{T,i}]\bigr)
\in
\mathcal{M}
:=
\Gr(r_X,X)
\times
\Gr(r_Y,Y)
\times
\Gr(r_T,T).
\end{equation}
A field with a third spatial axis adds one Grassmann factor to the product and one index to the core, and the construction is otherwise unchanged.

For comparison, the single-subspace representation reshapes the spatial dimensions into one axis,
\begin{equation}
\label{eq:flattened_snapshot}
M_i
=
\mathrm{reshape}\bigl(\mathcal{X}(\mu_i),[XY,T]\bigr)
\in
\mathbb{R}^{(XY)\times T},
\end{equation}
takes a rank-$r_s$ spatial basis $U_{s,i}\in\mathbb{R}^{(XY)\times r_s}$ of $M_i$, and represents it by the single point $[U_{s,i}]\in\Gr(r_s,XY)$. Its interpolation runs in one space of dimension $XY$, whereas the mode-wise representation runs in three spaces of dimensions $X$, $Y$ and $T$.

\subsection{Reference alignment of Tucker factors and cores}
\label{ssec:gauge}

Representing the factors by their subspaces removes rotations from the geometric description, but the Tucker core still depends on the basis chosen inside each subspace. Replacing the factors by $U_{X,i}Q_X$, $U_{Y,i}Q_Y$ and $U_{T,i}Q_T$ with orthogonal $Q_\bullet\in O(r_\bullet)$ leaves the three Grassmann points and the reconstructed tensor unchanged, and the core absorbs the inverse rotations. Part of the difference between cores computed at nearby parameters is then a difference of coordinates alone. To remove it, all training decompositions are placed in one reference frame before the cores are modeled.

The reference $\mathbf{U}_0=\bigl([U_{X,0}],[U_{Y,0}],[U_{T,0}]\bigr)\in\mathcal{M}$ is a mean of the training product points~\citep{karcher1977,afsari2011}. For each mode $\bullet\in\{X,Y,T\}$ one orthonormal representative $U_{\bullet,0}$ of the reference subspace is fixed and reused unchanged through alignment and through every query, since the aligned-core entries are coordinates relative to its columns. For training sample $i$, the rotation that brings its basis closest to the reference is the solution of the orthogonal Procrustes problem~\citep{Schoenemann1966,gower2004procrustes},
\begin{equation}
\label{eq:procrustes_problem}
Q_{\bullet,i}
=
\arg\min_{Q\in O(r_\bullet)}
\|U_{\bullet,i}Q-U_{\bullet,0}\|_F .
\end{equation}
With the singular value decomposition of the overlap matrix,
\begin{equation}
\label{eq:overlap_svd}
U_{\bullet,0}^\top U_{\bullet,i}
=
L_{\bullet,i}\Sigma_{\bullet,i}R_{\bullet,i}^\top ,
\end{equation}
the minimizer is its orthogonal part,
\begin{equation}
\label{eq:procrustes_rotation}
Q_{\bullet,i}=R_{\bullet,i}L_{\bullet,i}^\top ,
\qquad \bullet\in\{X,Y,T\}.
\end{equation}
The aligned factor is $\bar U_{\bullet,i}=U_{\bullet,i}Q_{\bullet,i}$. To keep the reconstructed tensor unchanged, the same rotations are transferred to the core,
\begin{equation}
\label{eq:gauge_aligned_core}
\bar{\mathcal{G}}_i
=
\mathcal{G}_i
\times_1 Q_{X,i}^\top
\times_2 Q_{Y,i}^\top
\times_3 Q_{T,i}^\top .
\end{equation}
This transformation changes only the coordinates, and differences among the aligned cores $\bar{\mathcal{G}}_i$ then reflect the parameter dependence and not the orientation of the Tucker factors.

The principal angles between two subspaces with orthonormal representatives $U$ and $\widetilde U$ are $\theta_j=\arccos\sigma_j(U^\top\widetilde U)$, reported throughout as multiples of $\pi/2$, the angle at which the Grassmann logarithm ceases to be single-valued. The rotation is unique when every angle between a training subspace and the reference is strictly below $\pi/2$, since the overlap $U_{\bullet,0}^{\top}U_{\bullet,i}$ is then nonsingular. Any reference meeting this condition supports the construction, and the mean is used because it keeps the training angles small. After alignment the cores are ordinary Euclidean tensors. The factor subspaces are represented by tangent coordinates, and these do not depend on the alignment because the logarithm map acts on subspaces.

\subsection{Tangent-space prediction of the factor subspaces}
\label{ssec:tangent}

The factor subspaces are interpolated in the tangent space at the reference $\mathbf{U}_0$, following standard Grassmann interpolation~\citep{Amsallem2008,Son2013,Bendokat2020}. Each training product point is mapped to that tangent space by the Grassmann logarithm, the initial velocity of the shortest geodesic from $\mathbf{U}_0$ to $\mathbf{U}_i$,
\begin{equation}
\label{eq:tangent_coords}
\mathbf{v}_i=\log_{\mathbf{U}_0}(\mathbf{U}_i).
\end{equation}
Because $\mathcal{M}$ is a Cartesian product, the logarithm acts on the three modes independently, and each mode-wise component is represented by its tangent coordinate, a matrix $Z_{\bullet,i}$ of the factor's size with $U_{\bullet,0}^\top Z_{\bullet,i}=0$. That condition removes rotations inside the reference subspace and keeps only the directions that tilt the subspace away from it. The closed form of the logarithm returns these representatives, collected as the tangent coordinates
\begin{equation}
\label{eq:tangent_triple}
\mathbf{Z}_i=\bigl(Z_{X,i},Z_{Y,i},Z_{T,i}\bigr),
\end{equation}
which the surrogate stores.

At a query parameter $\mu$ the tangent coordinates are predicted by a weighted average of the training coordinates,
\begin{equation}
\label{eq:tangent_regressor}
\hat{\mathbf{Z}}(\mu)
=
\sum_{i=1}^N w_i(\mu)\mathbf{Z}_i
=
\bigl(
\hat Z_X(\mu),\hat Z_Y(\mu),\hat Z_T(\mu)
\bigr),
\end{equation}
the tangent coordinate of the predicted tangent vector $\hat{\mathbf{v}}(\mu)$. One set of scalar weights serves all three modes and places the $X$-, $Y$- and $T$-mode predictions in one parameter-space neighborhood. The same weights drive the core prediction below. Each parameter coordinate is centered and scaled to unit variance with the training statistics, the $k$ nearest training parameters in that standardized space are selected, and a Gaussian kernel $w_i(\mu)\propto\exp(-d_i^2/(2h^2))$ over their distances $d_i$ is normalized to sum to one. The bandwidth $h$ is set per query to the larger of the standard deviation and half the mean of the $k$ distances, and training parameters outside the neighborhood receive zero weight.

The predicted tangent vector is a convex combination of the neighbors' tangent vectors. Its largest principal angle cannot exceed the largest angle among those neighbors, and the prediction lies where the logarithm is single-valued. As a numerical safeguard, the implementation rescales a prediction whose largest angle exceeds $0.9\,(\pi/2)$ back to that threshold before the exponential map, and the supplementary appendix reports how often this happens.

The exponential map returns the predicted tangent vector to the product manifold,
\begin{equation}
\label{eq:exp_prediction}
\hat{\mathbf{U}}(\mu)
=
\exp_{\mathbf{U}_0}\bigl(\hat{\mathbf{v}}(\mu)\bigr),
\end{equation}
again one mode at a time and through its closed form, and gives orthonormal representatives $\hat U_X(\mu)$, $\hat U_Y(\mu)$ and $\hat U_T(\mu)$ of the predicted subspaces.

\subsection{Low-rank prediction of the aligned core}
\label{ssec:prediction}

Because the aligned cores $\bar{\mathcal{G}}_i$ are all written in the common frame of the alignment step, their means, differences and principal components are comparable across parameter values. Principal component analysis of the vectorized residuals about the mean aligned core then gives a low-rank aligned-core model, in the spirit of the tensorial reduced-order modeling of \citet{Mamonov2022,Mamonov2024}.

The mean aligned core,
\begin{equation}
\label{eq:core_mean}
\bar{\mathcal{G}}_{\mathrm{mean}}=\frac{1}{N}\sum_{i=1}^N \bar{\mathcal{G}}_i ,
\end{equation}
holds the part of the core that does not vary with the parameter. Each residual is vectorized in a fixed entrywise order, $G(:,i)=\mathrm{vec}\bigl(\bar{\mathcal{G}}_i-\bar{\mathcal{G}}_{\mathrm{mean}}\bigr)$, which gives $G\in\mathbb{R}^{r_Xr_Yr_T\times N}$. A truncated singular value decomposition $G\approx U_p\Sigma_pV_p^\top$ keeps the $r_p$ leading residual-core directions, and reshaping the columns of $U_p\Sigma_p$ into their tensor shape gives
\begin{equation}
\label{eq:coc_basis}
\mathcal{G}_{\mathrm{basis}}
=
\mathrm{reshape}\bigl(U_p\Sigma_p,\,[r_X,r_Y,r_T,r_p]\bigr)
\in\mathbb{R}^{r_X\times r_Y\times r_T\times r_p},
\end{equation}
whose fourth index runs over the retained patterns, shared by every training run. Row $i$ of $V_p$ holds the coordinates of run $i$ in those patterns, and the parameter dependence lives in these coordinates. The rank $r_p$ sets the compression of the core model.

At a query only the coordinates are predicted, with the pattern basis held fixed. The weights $w_i(\mu)$ of \eqref{eq:tangent_regressor}, collected into $w(\mu)\in\mathbb{R}^N$ and nonzero on the $k$ selected neighbors, give the coordinate vector $a(\mu)=V_p^\top w(\mu)\in\mathbb{R}^{r_p}$, and the predicted aligned core is
\begin{equation}
\label{eq:online_core}
\bar{\mathcal{G}}(\mu) = \bar{\mathcal{G}}_{\mathrm{mean}} + \mathcal{G}_{\mathrm{basis}}\bar{\times}_4 a(\mu),
\end{equation}
where $\bar{\times}_4$ contracts the fourth index of $\mathcal{G}_{\mathrm{basis}}$ against $a(\mu)$ and returns an $r_X\times r_Y\times r_T$ tensor.

\subsection{Compatible reconstruction}
\label{ssec:reconstruction}

The core holds coefficients relative to particular factor bases, and the predicted factors and the predicted core have to share one basis orientation before reconstruction. The alignment step fixes that orientation for the training data, and the Grassmann exponential map returns the predicted factors in the same reference orientation. The aligned cores behind \eqref{eq:core_mean} to \eqref{eq:online_core} are therefore written in the coordinates of the predicted factors, and no further alignment is needed at a query parameter. The supplementary appendix proves this compatibility and shows that the alignment step introduces no approximation error. With $\bar U_{\bullet}(\mu):=\hat U_{\bullet}(\mu)$ the field is
\begin{equation}
\label{eq:final_reconstruction}
\hat{\mathcal{X}}(\mu)
=
\bar{\mathcal{G}}(\mu)
\times_1 \bar U_X(\mu)
\times_2 \bar U_Y(\mu)
\times_3 \bar U_T(\mu).
\end{equation}
At fixed Tucker ranks, three sources of error remain in $\hat{\mathcal{X}}(\mu)$: the interpolation of the tangent coordinates in $\hat{\mathbf{Z}}(\mu)$, the interpolation of the core coordinates in $a(\mu)$ and the residual-core truncation at rank $r_p$. The Tucker truncation of the field itself is a separate representation error, measured beside the prediction error in the experiments.

\subsection{Computational complexity}
\label{ssec:complexity}

A cached query runs in two stages. The Grassmann stage forms the weighted sum of the neighbors' tangent coordinates and applies the exponential map. The reconstruction stage then assembles the full $X\times Y\times T$ field. The two methods differ mainly in the Grassmann stage. The single-subspace method works on one basis of ambient dimension $XY$. The product-Grassmann method works on three bases of ambient dimensions $X$, $Y$ and $T$. Reconstruction produces the same field in both methods, at a cost of the same order. Since the query time is the sum of the two stages, the speedup of the whole query lies between the speedup of the Grassmann stage and that of the reconstruction, and the shared reconstruction therefore bounds the end-to-end gain. Table~\ref{tab:complexity_side} lists the leading cost of each stage and the values stored per training sample. Its interpolation-size row counts only the values interpolated in the Grassmann stage, without the core or coefficient terms, which favors the single-subspace method.

\begin{table}[!htb]
\centering
\caption{Online cost per query and per-sample storage of the two surrogates.}
\label{tab:complexity_side}
\footnotesize
\setlength{\tabcolsep}{3pt}
\begin{tabular}{lcc}
\toprule
Stage & Product-Grassmann & Single-subspace \\
\midrule
Kernel weights & $\mathcal{O}(Nd)$ & $\mathcal{O}(Nd)$ \\
Tangent interpolation & $\mathcal{O}\bigl(k(Xr_X{+}Yr_Y{+}Tr_T)\bigr)$ & $\mathcal{O}(kXYr_s)$ \\
Exponential map & $\mathcal{O}(Xr_X^2{+}Yr_Y^2{+}Tr_T^2)$ & $\mathcal{O}(XYr_s^2)$ \\
Core / temporal coefficients & $\mathcal{O}(kr_p{+}r_Xr_Yr_Tr_p)$ & $\mathcal{O}(kr_sT)$ \\
Reconstruction & $\mathcal{O}(Xr_Xr_Yr_T{+}XYr_Yr_T{+}XYTr_T)$ & $\mathcal{O}(XYTr_s)$ \\
\midrule
Interpolation size (values) & $Xr_X+Yr_Y+Tr_T$ & $XYr_s$ \\
Stored per sample (values) & $Xr_X+Yr_Y+Tr_T+r_p$ & $XYr_s+r_sT$ \\
\bottomrule
\end{tabular}
\end{table}

The main saving lies in the interpolation size. At comparable ranks the single-subspace method interpolates $XYr_s$ values and the product-Grassmann method $Xr_X+Yr_Y+Tr_T$, a ratio of about $XY/(X+Y+T)$. The grid sets the scale of this ratio, and the ranks selected on training data determine its realized value. A third spatial axis increases it, since the single-subspace size grows with the product of the spatial sizes and the mode sum only with their total. The stored values scale the same way. The closed-form logarithm and exponential maps, the proofs and the step-by-step offline and online algorithms are in the supplementary appendix.

\section{Numerical Experiments}
\label{sec:experiments}

Five parameterized PDE benchmarks compare the proposed product-Grassmann Tucker surrogate, M1, with a controlled single-subspace Grassmannian surrogate, M2, on the same data and with the same local interpolation. The experiments address whether the smaller mode-wise representation preserves prediction accuracy and, where the two methods differ, how the difference relates to compression and to prediction. Two ablations, M1-Fixed and M1-Unaligned, isolate the roles of core variation and alignment. The supplementary appendix collects the numerical settings of every benchmark, the selection and resampling studies, and the robustness checks.

\subsection{Experimental protocol}
\label{ssec:experiment_protocol}

The methods compared are summarized in Table~\ref{tab:methods}. M2 is a single-subspace Grassmannian surrogate built on the same training data, local interpolation weights and Grassmannian prediction procedure as M1. The two methods therefore share the parameter-space interpolation and differ in the field representation and in the coefficients they predict. M2 represents the field by the spatial basis $U_{s,i}$ of the reshaped matrix \eqref{eq:flattened_snapshot}, one point on $\Gr(r_s,XY)$, and interpolates it exactly as M1 interpolates each mode-wise subspace, through the tangent space at a reference mean $[U_{s,0}]$ with the same kernel weights. Its coefficients follow the same reference convention as M1's cores. Each training basis is rotated onto the reference by Procrustes, giving $\bar U_{s,i}$; the temporal coefficients are taken in that aligned frame, $B_i=\bar U_{s,i}^\top M_i\in\mathbb{R}^{r_s\times T}$, the full reduced time history with one column per stored step; and the reconstruction is $\hat U_s(\mu)\hat B(\mu)$. M2 keeps the full time resolution, whereas M1 compresses time to rank $r_T$, and the supplementary appendix quantifies this temporal asymmetry on the wave and heat cases. The spatial basis of M2 comes from a full singular value decomposition, and the retained rank is the one the selection rule chooses, up to the matrix-rank ceiling $\min(XY,T)$.

The two ablations change one element of M1 each. M1-Fixed keeps one core for all parameters, the aligned training mean, as the single decomposition of \citet{sasagawa2023tucker} keeps one core. The factor subspaces of M1-Fixed are predicted with the same step and weights as M1, whereas the method of Sasagawa and Tanaka carries the parameter dependence in a spline-interpolated parameter-mode factor.

\begin{table}[!htb]
\centering
\caption{Methods compared. M1-Fixed and M1-Unaligned are ablations of M1.}
\label{tab:methods}
\footnotesize
\setlength{\tabcolsep}{3pt}
\begin{tabular}{@{}l>{\raggedright\arraybackslash}p{0.30\linewidth}>{\raggedright\arraybackslash}p{0.27\linewidth}>{\raggedright\arraybackslash}p{0.22\linewidth}@{}}
\toprule
Method & Subspace representation & Core treatment & Role \\
\midrule
M1 & mode-wise product Grassmann & low-rank aligned-core model & proposed surrogate \\
M2 & single Grassmann subspace after reshaping the spatial dimensions & parameter-dependent reduced coefficients & representation comparison \\
\midrule
M1-Fixed & mode-wise product Grassmann & fixed mean aligned core & core-variation ablation \\
M1-Unaligned & mode-wise product Grassmann & low-rank core model without alignment & alignment ablation \\
\bottomrule
\end{tabular}
\end{table}

Each method is tuned on its own training data. The Tucker ranks of M1, the single-subspace rank $r_s$ of M2, the neighborhood size $k$ and the core rank $r_p$ are selected by validation on the training data~\citep{arlot2010survey}. Each candidate is scored by its mean relative error over held-out training samples, with leave-one-out throughout except for the larger single-subspace search of Case~V, which uses ten training folds. On the rank ladders, candidates whose score lies within a tenth of a percentage point of the best count as tied, and the tied candidate with the smallest interpolation size is taken. The parameters $k$ and $r_p$ do not set the interpolation size and take the best score, with ties within a hundredth of a point resolved to the smaller value. The tenth-of-a-point tie band was adopted after an exploratory comparison of two candidate bands on the original splits. For that reason the supplementary appendix also reports a nested resampling study that re-selects every model choice inside each training split under both bands. The candidate ladders, fixed in advance, and the remaining details of the selection are in the supplementary appendix and the repository.

Prediction accuracy is measured on held-out test parameters by the relative Frobenius error,
\begin{equation}
\label{eq:rel_error}
\epsilon_{\mathrm{rel}}(\mu)
=
\frac{\|\mathcal{X}_{\mathrm{true}}(\mu)-\hat{\mathcal{X}}(\mu)\|_F}
{\|\mathcal{X}_{\mathrm{true}}(\mu)\|_F},
\end{equation}
on the original split of each benchmark and on ten resampled splits, and subspace diagnostics use the mode-wise principal angles. The field figures of the case studies show the mid-time frame of the worst test sample, the held-out sample with the largest M1 error. Every reported error is computed over the complete spatiotemporal tensor. M1 aims to retain the accuracy of M2 at lower cost. The accuracy comparison therefore refits both methods on each resampled split, at their selected ranks and hyperparameters, and reports the paired M1 minus M2 difference in mean error with a two-sided $95\%$ interval. The difference counts as within tolerance when the upper limit of that interval stays below a prespecified tolerance of one percentage point. The supplementary appendix repeats the comparison against a relative margin of ten percent of the M2 error and, since the resampled test sets overlap, with a cluster bootstrap over simulations.

Cost is measured in two ways. The interpolation size is the number of values the Grassmann stage interpolates, $Xr_X+Yr_Y+Tr_T$ for M1 and its ablations against $XYr_s$ for M2. The query time is the wall-clock time of a fully cached query, each method summing only the $k$ active neighbors, and it includes the full-field reconstruction present in both methods. The kernel weights, the same for both methods, are formed before the timed region. Both surrogates already replace the solver at query time, and every speedup reported below is the further saving of M1 over the single-subspace surrogate.

Each case study also follows both methods across a common rank budget $r$, the sum of the Tucker ranks for M1 and the rank $r_s$ for M2. The budget indexes model size without equating storage or cost. At every budget from $4$ to that of the reported ranks, and at least to $15$, the search fits each combination of Tucker ranks with that sum from full mode-wise SVDs. It keeps the one with the lowest ten-fold cross-validation error on the training samples of the original split, and the test samples take no part in this choice. A budget with more than $150$ combinations is searched within three of the previous budget's choice in every mode. The kept combination and M2 at $r_s=r$ are then refitted on the ten resampled splits with $k$ and $r_p$ at their selected values. On Case~IV the single-subspace rank cannot exceed the matrix-rank ceiling of the benchmark, and the M2 curve ends there. The case-study figures plot the mean test error over those splits, the interpolation size and the cached query time against $r$, with both methods timed together in one session. The reported ranks remain those of the leave-one-out rule.

A method's error combines the information lost when the field is compressed and the error added when the reduced representation is predicted at a new parameter. To separate them, the reconstruction error is bounded by
\begin{equation}
\label{eq:error_decomp}
\|\mathcal{X}(\mu)-\hat{\mathcal{X}}(\mu)\|_F
\le
\|\mathcal{X}(\mu)-\mathcal{X}_{\mathrm{repr}}(\mu)\|_F
+
\|\mathcal{X}_{\mathrm{repr}}(\mu)-\hat{\mathcal{X}}(\mu)\|_F,
\end{equation}
where $\mathcal{X}_{\mathrm{repr}}(\mu)$ is the reconstruction of the true field in the method's own representation at the selected ranks, computed with the decomposition procedure used for that benchmark. The first term is the representation error, which remains even when the surrogate recovers the chosen factors and core exactly at $\mu$. The second is the prediction error, the error introduced by predicting the subspaces and the core at $\mu$ in place of recomputing them, and it absorbs the rank-$r_p$ core-model truncation. Both terms are evaluated on the held-out samples. For the matricized field of M2 the representation error at rank $r_s$ is the exact singular value tail. For the Tucker format it is the projection error at the selected ranks, and the supplementary appendix summarizes the standard HOSVD bounds that relate it to the best Tucker approximation~\citep{DeLathauwer2000a,Vannieuwenhoven2012,Hackbusch2012}.

\subsection{Benchmark problems}
\label{ssec:benchmarks}

The numerical study uses five parameterized PDE benchmarks that differ in how compressible the field is along its modes, how far its subspaces move with the parameter and how many spatial modes it carries. Four are constructed for this study, and Case~IV draws on the public CFDBench collection. Table~\ref{tab:benchmarks} summarizes the governing problem, the grid, the parameter dimension $d$, the original split and the role of each benchmark.

\begin{table}[!htb]
\centering
\caption{The five benchmarks: grid (spatial resolution $\times$ stored times), parameter dimension $d$, original training/test split and purpose.}
\label{tab:benchmarks}
{\footnotesize\setlength{\tabcolsep}{3pt}%
\begin{tabular}{@{}c >{\raggedright\arraybackslash}p{0.29\linewidth} ccc >{\raggedright\arraybackslash}p{0.31\linewidth}@{}}
\toprule
Case & System & Grid & $d$ & Original split & Purpose \\
\midrule
I & Anisotropic diffusion-reaction & $256^2{\times}100$ & 3 & 30/10 & smooth field, easy compression \\
II & Double-slit wave interference & $81^2{\times}500$ & 2 & 30/6 & moving interference, hard prediction \\
III & Multimodal forced heat & $128^2{\times}101$ & 2 & 30/10 & variation concentrated in the temporal mode \\
IV & Tube flow (CFDBench), external & $64^2{\times}15$ & 2 & 78/20 & public data, single-subspace representation exact \\
V & Advection-diffusion, two-scale & $32^3{\times}24$ & 3 & 40/10 & third spatial mode \\
\bottomrule
\end{tabular}}%
\end{table}

\subsection{Case Study I: anisotropic diffusion-reaction}
\label{sec:case_anisotropic}

The first benchmark is a smooth field that compresses easily. Little information is lost at the selected ranks, and the subspaces vary gently with the parameter. Nearly all of the error therefore comes from predicting the reduced representation at a new parameter, and the case compares the two methods almost entirely on prediction. The benchmark is a two-dimensional anisotropic diffusion-reaction equation on the unit square,
\begin{equation}
\frac{\partial u}{\partial t}=D_x\frac{\partial^2u}{\partial x^2}+D_y\frac{\partial^2u}{\partial y^2}-\lambda u+g(x,y)q(t),
\qquad (x,y)\in[0,1]^2,\ t\in(0,1],
\label{eq:case1_pde}
\end{equation}
with zero initial condition, periodic boundaries and a separable Gaussian source. The parameter vector $\mu=(D_x,D_y,\lambda)$ sets the diffusion in the two spatial directions and the decay rate.

The $40$ samples are a Latin hypercube over $[0.003,0.008]\times[0.002,0.005]\times[0.40,0.90]$; $30$ are used for training and $10$ interior points, each at least $0.15$ of every parameter range from the boundary of the box, are held out for testing. The fields are computed by a Fourier spectral method on a $256\times256$ grid and stored at $100$ times over $t\in[0,1]$.

At the selected ranks the two methods reach the same mean test error. Holding the core at its training mean more than triples that error. Even on a field that compresses easily, the core varies with the parameter.

Both methods recover the worst test sample closely, with residuals of the same size and shape (Figure~\ref{fig:case1_fields_v6}).

\begin{figure}[!htb]
  \centering
  \includegraphics[width=\fieldwidth]{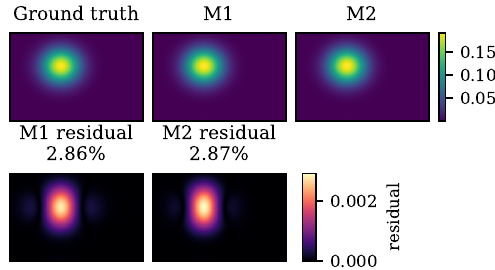}
  \caption{Case Study~I, worst test sample, mid-time frame: truth, M1 and M2 on one color scale, and the two residuals on another.}
  \label{fig:case1_fields_v6}
\end{figure}

M1 attains the best mean error of M2, $1.18$ percent, to within $0.01$ points already at a budget of $r=6$, with ranks $(2,2,2)$ and $107$ times fewer interpolated values (Figure~\ref{fig:budget_I}). On this easily compressed field both error curves saturate at small budgets, and M2 lies within $0.01$ points of its best mean from $r_s=2$ onward.

\begin{figure}[!htb]
  \centering
  \includegraphics[width=\linewidth]{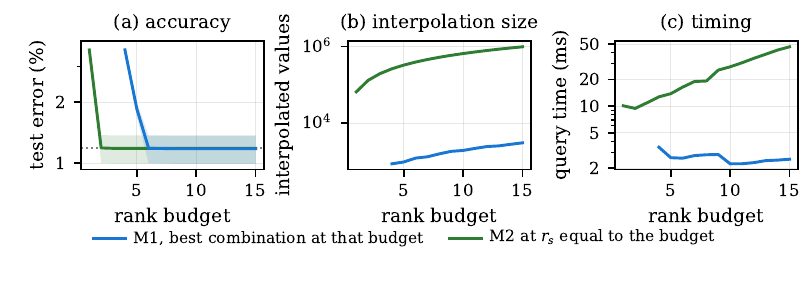}
  \caption{Case Study~I against the rank budget $r$: (a) mean test error over ten resampled splits with one standard deviation shaded and the best mean error of M2 dotted, (b) interpolation size and (c) cached query time.}
  \label{fig:budget_I}
\end{figure}

\subsection{Case Study II: double-slit wave interference}
\label{sec:case_doubleslit}

Moving interference patterns make the wave benchmark hard to predict. A small change in frequency shifts the pattern across the domain, and a reconstruction that looks reasonable can still carry a large error. The model is wave propagation on the unit square,
\begin{equation}
\frac{\partial^2u}{\partial t^2}=\Delta u,\qquad (x,y)\in[0,1]^2,\ t\in(0,1],
\label{eq:case2_pde}
\end{equation}
at unit wave speed, driven through two slits on the left boundary by a sinusoid with amplitude $\mu_1$ and frequency parameter $\mu_2$, which form $\mu=(\mu_1,\mu_2)$.

The $36$ runs form the full $6\times6$ grid of $\mu_1\in\{60,72,84,96,108,120\}$ against $\mu_2\in\{3,3.5,4,4.5,5,5.5\}$; $30$ are used for training and $6$ interior grid points are held out for testing. Each run is stored on an $81\times81$ grid at $500$ time steps. The archive does not record the integration scheme, the slit positions or the treatment of the other boundaries.

On this field the two methods are nearly identical in accuracy, and a fixed core gives three times their error. Both methods leave the same narrow residual bands (Figure~\ref{fig:case2_fields_v6}), consistent with small errors in the predicted phase that displace the interference bands.

\begin{figure}[!htb]
  \centering
  \includegraphics[width=\fieldwidth]{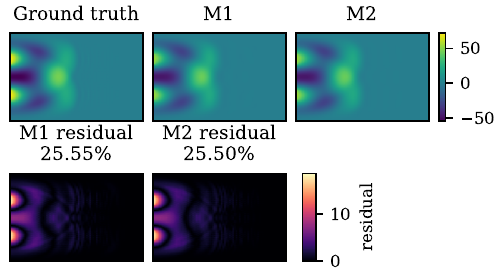}
  \caption{Case Study~II, worst test sample, mid-time frame: truth, M1 and M2 on one color scale, and the two residuals on another.}
  \label{fig:case2_fields_v6}
\end{figure}

On the wave field M1 remains $0.09$ points above the best mean error of M2, $22.28$ percent, at the largest budget considered, $r=43$, the budget of its selected ranks, and interpolates $19$ times fewer values than M2 at $r_s=32$ (Figure~\ref{fig:budget_II}). Its error is still decreasing at that budget.

\begin{figure}[!htb]
  \centering
  \includegraphics[width=\linewidth]{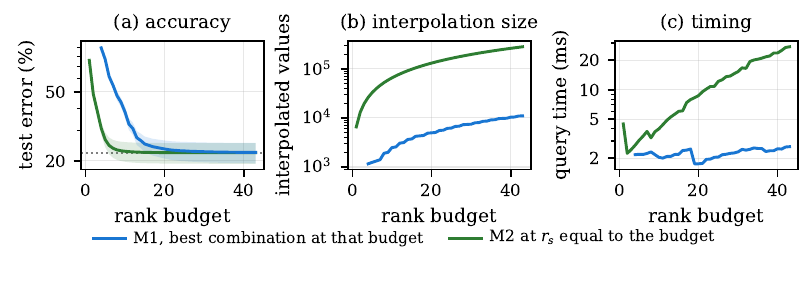}
  \caption{Case Study~II against the rank budget $r$, panels as in Figure~\ref{fig:budget_I}.}
  \label{fig:budget_II}
\end{figure}

\subsection{Case Study III: multimodal forced heat equation}
\label{sec:case_heat}

In the heat benchmark the parametric variation is concentrated in the temporal mode. The model is a forced diffusion equation on the unit square with homogeneous Dirichlet boundary conditions,
\begin{equation}
\frac{\partial u}{\partial t}=D(\mu_2)\Delta u+S(x,y,t;\mu),\qquad D(\mu_2)=D_0(1+0.5\mu_2),\ D_0=0.003,
\label{eq:case3_pde}
\end{equation}
with parameter vector $\mu=(\mu_1,\mu_2)\in[0.5,1.5]^2$. The initial condition mixes eight sine modes, and the forcing is a sum of nine Gaussian-in-time pulses on sine modes, all with spatial index at most three. The amplitudes of the modes and pulses depend on the parameter, and the diffusivity on $\mu_2$. The spatial sine content is fixed, and the parameter changes the temporal behavior.

The $40$ samples are independent uniform draws over the square under a fixed seed; the first $30$ are used for training and the last $10$ are held out for testing. Each run is computed on a $128\times128$ grid and stored at $101$ times over $t\in[0,0.4]$.

The selected ranks follow the structure of the field. Diffusion damps each Fourier mode without mixing it. The spatial subspaces are then the same across the training set, and the temporal subspaces turn with the parameter. Both spatial ranks of M1 are three, the exact rank of the spatial content, and only the temporal rank rises, to four. The M2 rank $r_s=9$ is the product of the two spatial ranks, and the two methods therefore carry the same spatial content.

The two methods are close in accuracy, with almost equal prediction errors. The small difference comes from the temporal compression of M1. Freezing the core increases the M1 error more than fivefold.

The worst test sample leaves a smooth, domain-wide residual (Figure~\ref{fig:case3_fields_v6}), consistent with an error in the predicted temporal behavior.

\begin{figure}[!htb]
  \centering
  \includegraphics[width=\fieldwidth]{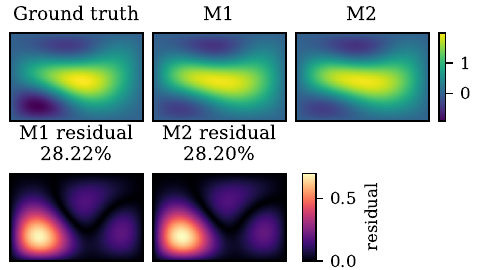}
  \caption{Case Study~III, worst test sample, mid-time frame: truth, M1 and M2 on one color scale, and the two residuals on another.}
  \label{fig:case3_fields_v6}
\end{figure}

M1 attains the best mean error of M2, $7.82$ percent, to within $0.01$ points at $r=13$ with $100$ times fewer interpolated values than M2 at $r_s=9$ (Figure~\ref{fig:budget_III}). From $r=7$ onward the cross-validated combination holds both spatial ranks at three and assigns the remaining budget to the temporal mode. At $r=10$, the budget of the selected ranks, it coincides with the selection $(3,3,4)$.

\begin{figure}[!htb]
  \centering
  \includegraphics[width=\linewidth]{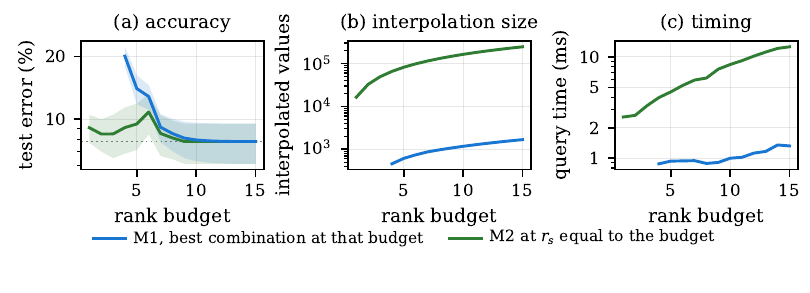}
  \caption{Case Study~III against the rank budget $r$, panels as in Figure~\ref{fig:budget_I}.}
  \label{fig:budget_III}
\end{figure}

\subsection{Case Study IV: tube flow (CFDBench)}
\label{sec:case_cfdbench}

Case~IV applies the method to external data, the \texttt{tube} problem of CFDBench~\citep{luo2023cfdbench}, a public collection of parametric computational fluid dynamics runs~\citep{cfdbench2023data}, in a setting where the single-subspace representation is exact and the error of M2 is prediction error alone. The benchmark is laminar tube flow, governed by the incompressible Navier--Stokes equations,
\begin{equation}
\rho\left(\frac{\partial\vect{u}}{\partial t}+\vect{u}\cdot\nabla\vect{u}\right)=-\nabla p+\eta\,\nabla^2\vect{u},
\qquad \nabla\cdot\vect{u}=0,
\label{eq:case4_pde}
\end{equation}
integrated by a finite-volume solver and interpolated by the collection's authors onto a $64\times64$ grid. The subset used varies the fluid properties at fixed geometry and inlet velocity, and the parameter is the pair $\mu=(\mu_1,\mu_2)$, the density index of the collection, used as an ordinal coordinate, and the dynamic viscosity $\eta$.

The subset provides a $10\times10$ parameter grid; two runs with fewer than $15$ stored frames are removed, leaving $98$, of which $78$ are used for training and $20$, drawn at random under a fixed seed, are held out for testing, $8$ of them on the boundary of the grid. Each run stores the velocity magnitude on a $64\times64$ grid over $15$ time steps, and the full provenance is in the supplementary appendix.

M1 truncates its two spatial modes at rank $20$ each and keeps the full temporal rank. Each run stores only $15$ time steps, and the matricized field therefore has rank at most $\min(XY,T)=15$, which M2 retains.

M2 is slightly more accurate here. The prediction-error terms of the two methods almost coincide, and the gap is consistent with the Tucker representation error of M1. With a fixed core the error is almost five times that of M1.

The M1 residual of the worst test sample is spread across the interior of the field (Figure~\ref{fig:case4_fields}).

\begin{figure}[!htb]
  \centering
  \includegraphics[width=\fieldwidth]{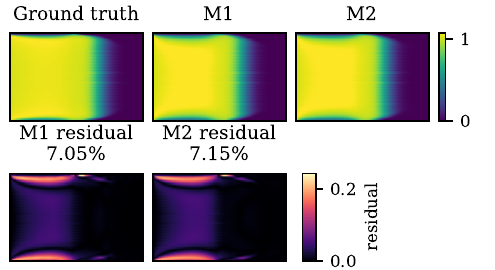}
  \caption{Case Study~IV, worst test sample, mid-time frame: truth, M1 and M2 on one color scale, and the two residuals on another.}
  \label{fig:case4_fields}
\end{figure}

On this external benchmark M1 attains the error of M2 to within $0.01$ points at $r=67$ with $17$ times fewer interpolated values, and a larger core rank takes it below M2 (Figure~\ref{fig:budget_IV}). M2 improves with the budget up to its ceiling $r_s=15$, where its representation is exact, and ends at $1.76$ percent. At $r=55$, the budget of the selected ranks, the cross-validated combination $(13,27,15)$ attains $1.79$ percent, below the $1.88$ percent of the selected $(20,20,15)$, whose ladder constrains the two spatial ranks to be equal. The sweep ends at $r=80$, where over the last fifteen budgets each unit changes the error by less than $0.002$ points. Raising $r_p$ there from the selected $10$ to $16$ lowers the M1 error from $1.76$ to $1.72$ percent, below M2.

\begin{figure}[!htb]
  \centering
  \includegraphics[width=\linewidth]{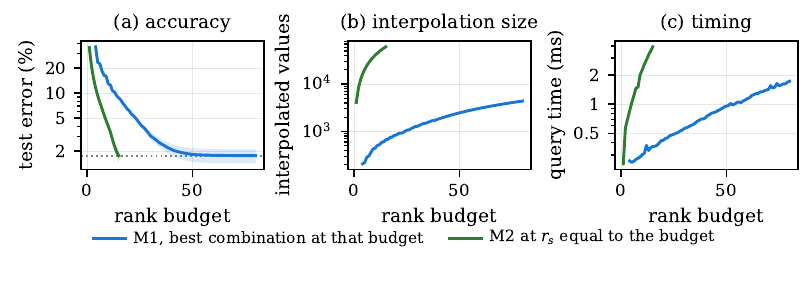}
  \caption{Case Study~IV against the rank budget $r$, panels as in Figure~\ref{fig:budget_I}.}
  \label{fig:budget_IV}
\end{figure}

\subsection{Case Study V: three-dimensional advection-diffusion}
\label{sec:case_3d}

The fifth benchmark adds a third spatial mode $Z$, on a field with two spatial scales. The single-subspace object then grows with the product $XYZ$ of the three spatial sizes and the Tucker object with their sum. The cost model predicts its largest saving here. The scalar field solves advection-diffusion on the periodic cube $\mathbb{T}^3=[0,1)^3$,
\begin{equation}
\partial_t u + c\,\partial_x u = \nu\,\nabla^2 u,
\qquad c=0.48,\quad \nu=2.5\cdot10^{-3},\quad t\in[0,1.25],
\label{eq:case3d_pde}
\end{equation}
with $c$ the advection speed along $x$ and $\nu$ the diffusivity. Because the Fourier propagator evaluates the solution directly at the stored times, no time-stepping error enters the generated fields. The parameter $\mu=(\mu_y,\mu_z,\alpha)$ enters through the initial condition, which superposes a narrow Gaussian pulse of width $0.15$ centered at $(0.14,\mu_y,\mu_z)$ and $\alpha$ times a wide Gaussian of width $0.28$ at a fixed center, both in the periodic distance. The first two parameters move the narrow pulse across the transverse plane and so turn the mode-$y$ and mode-$z$ subspaces. The third sets the balance of the two scales, which decay at different rates, and thereby changes the temporal signature and with it the core.

The $50$ samples are a Latin hypercube over $[0.40,0.60]^2\times[0.15,1.00]$ under a fixed seed; $40$ are used for training and $10$ interior points, at least a tenth of each range from the edges, are held out for testing. Each run is stored on a $32\times32\times32$ grid over $24$ time steps.

The third spatial mode gives the largest interpolation reduction and query speedup of the five benchmarks, with M1 close to M2 in accuracy. The resampled gap, $+0.22$ points with a paired interval of $[+0.16,+0.28]$, is the largest of the five and inside the prespecified tolerance. The fixed core raises the error by a factor of $2.8$ and the unaligned variant, at $17.37$ percent, by a factor of $2.5$. The aligned core therefore varies with the parameter, and the Procrustes step makes the cores comparable.

Both methods place the narrow pulse and the broad background of the worst test sample (Figure~\ref{fig:case5_fields}). The residual comes from a visibly displaced pulse center.

\begin{figure}[!htb]
  \centering
  \includegraphics[width=\fieldwidth]{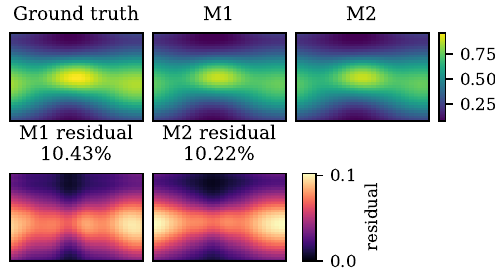}
  \caption{Case Study~V, worst test sample, mid-time frame on the $z$ slice through the pulse: truth, M1 and M2 on one color scale, and the two residuals on another.}
  \label{fig:case5_fields}
\end{figure}

At the largest budget considered, $r=15$, M1 remains $0.21$ points above the best mean error of M2 with $447$ times fewer interpolated values, and its error is still decreasing (Figure~\ref{fig:budget_V}). With three spatial modes, M2 lies within $0.01$ points of its best mean, $6.81$ percent, from $r_s=6$ onward. At $r=14$, the budget of the selected ranks, the cross-validated combination coincides with the selection $(5,2,2,5)$.

\begin{figure}[!htb]
  \centering
  \includegraphics[width=\linewidth]{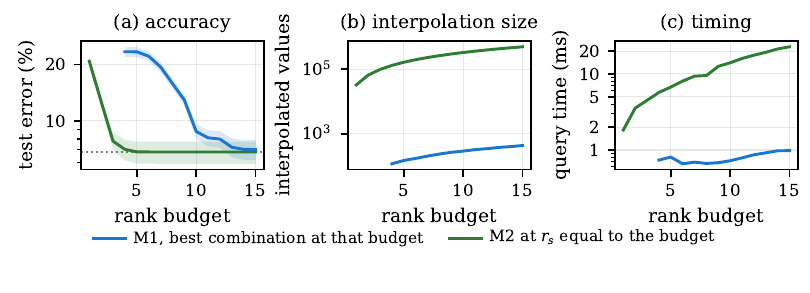}
  \caption{Case Study~V against the rank budget $r$, panels as in Figure~\ref{fig:budget_I}.}
  \label{fig:budget_V}
\end{figure}

On finer grids the interpolation advantage grows strongly. The wall-clock speedup grows more slowly because the full-field reconstruction takes an increasing share of the M1 query (Table~\ref{tab:case3d_resolution}). The same problem is repeated at $48^3$ and $64^3$ with the training design and the selected ranks unchanged, and the accuracy difference does not change.
\begin{table}[!htb]
\centering
\caption{Case Study~V at three grid resolutions, with the gap as mean $\pm$ standard deviation over ten resampled splits.}
\label{tab:case3d_resolution}
\footnotesize
\setlength{\tabcolsep}{3pt}
\begin{tabular}{rrrrrrrrr}
\toprule
$n$ & Points & M1 values & M2 values & Ratio & M1 (ms) & M2 (ms) & Speedup & Gap (pp) \\
\midrule
32 & 32{,}768 & 408 & 196{,}608 & $481.9\times$ & 1.0 & 8.2 & $8.5\times$ & $+0.22\pm0.09$ \\
48 & 110{,}592 & 552 & 663{,}552 & $1202.1\times$ & 2.2 & 24.8 & $11.1\times$ & $+0.22\pm0.09$ \\
64 & 262{,}144 & 696 & 1{,}572{,}864 & $2259.9\times$ & 4.9 & 57.9 & $11.9\times$ & $+0.22\pm0.09$ \\
\bottomrule
\end{tabular}
\end{table}

\subsection{Accuracy and cost across the benchmarks}
\label{ssec:cross_case_findings}

Across the five benchmarks M1 matches the accuracy of the single-subspace surrogate to within the prespecified one-percentage-point tolerance while reducing the size of the interpolation problem by factors of $13$ to $482$ (Table~\ref{tab:results_summary}). On Cases~II to~V the intervals exclude zero, and M2 is slightly more accurate there, by $0.07$ to $0.22$ points. Every upper limit is also below the relative margin of ten percent of the M2 error. Both ablations of M1 are worse than M1 on every case, the fixed core by factors of $2.8$ to $5.4$ and the unaligned core by factors of $2.3$ to $10.6$, with the per-case gaps and errors in the supplementary appendix. The largest reductions in interpolation size occur where the field admits low mode-wise ranks relative to the vectorized spatial dimension, most of all on the three-dimensional benchmark. The smallest is on the wave problem, whose $500$ stored time steps make the temporal mode a large part of the Tucker representation.

\begin{table}[!htb]
\centering
\caption{Accuracy and interpolation size at the selected ranks on the original test split, with errors in percent and the interpolation ratio M2 over M1.}
\label{tab:results_summary}
\footnotesize
\setlength{\tabcolsep}{4pt}
\begin{tabular}{llrrrr}
\toprule
Case & M1 ranks & $r_s$ & M1 & M2 & Interp.\ ratio \\
\midrule
I & $(2,4,4)$ & $5$ & 1.01 & 1.01 & $169.3\times$ \\
II & $(11,16,16)$ & $20$ & 21.38 & 21.34 & $12.9\times$ \\
III & $(3,3,4)$ & $9$ & 6.57 & 6.43 & $125.8\times$ \\
IV & $(20,20,15)$ & $15$ & 1.40 & 1.24 & $22.1\times$ \\
V & $(5,2,2,5)$ & $6$ & 6.90 & 6.69 & $481.9\times$ \\
\bottomrule
\end{tabular}
\end{table}

The cached query time falls by a smaller factor than the interpolation size, as the cost model of Section~\ref{ssec:complexity} predicts. Table~\ref{tab:cost_stages} compares the two stages with that prediction. Its stored-model column counts the reference basis once, the tangent coordinates and coefficients of every training sample and, for M1, the core model. The Grassmann stage is faster for M1 on every benchmark. The reconstruction stage is slower for M1 except on Case~III, with ratios of $0.35$ to $0.65$. Full-field reconstruction remains necessary for both methods and limits the end-to-end gain. The measured Grassmann-stage ratios run from $3.9$ to $28.1$, while the complete cached-query speedups run from $3.3$ to $8.5$.

\begin{table}[!htb]
\centering
\caption{Online cost ratios, M2 over M1, of the Grassmann stage, the reconstruction, the complete cached query and the stored model.}
\label{tab:cost_stages}
\footnotesize
\setlength{\tabcolsep}{6pt}
\begin{tabular}{lrrrr}
\toprule
& \multicolumn{3}{c}{Cached time} & \\
\cmidrule(lr){2-4}
Case & Grassmann & Recon. & Query & Stored \\
\midrule
I & 28.1 & 0.65 & $5.3\times$ & $169\times$ \\
II & 7.5 & 0.45 & $3.5\times$ & $13\times$ \\
III & 18.6 & 1.51 & $7.3\times$ & $125\times$ \\
IV & 3.9 & 0.41 & $3.3\times$ & $17\times$ \\
V & 19.9 & 0.35 & $8.5\times$ & $464\times$ \\
\bottomrule
\end{tabular}
\end{table}

Reconstruction limits the speedup on Case~I, and the Grassmann stage drives it on Case~V. On Case~I the Grassmann stage is much cheaper for M1, and reconstruction then accounts for most of its query cost. On Case~V the larger vectorized spatial dimension makes the Grassmann stage dominant for M2, and the speedup is the largest of the five. Offline training takes longer for M1 on Cases~II, III and~V, and the cached-query saving recovers that difference after about $1{,}500$ queries on Case~III and fewer on the other two, with M1 ahead from the first query on Cases~I and~IV.

The nested study of the supplementary appendix re-selects every model choice inside each resampled split. Under both tie bands every upper limit remains below the tolerance, although the mean gap on Case~II rises to $0.36$ points under the reported band. The boundary and noise checks lead to the same conclusion. The gap grows with the parameter dimension on Case~III, to $0.98$ points at six parameters, and widening the Case~V design box from $0.20$ to $0.40$ raises it to $5.30$ points.

Across the rank budget, M1 attains the best mean error of M2 to within $0.01$ points on Cases~I, III and~IV, and it remains $0.09$ and $0.21$ points above it on Cases~II and~V at the largest budgets considered. At the matching points of Cases~I, III and~IV, M1 interpolates $17$ to $107$ times fewer values and answers a cached query $2.9$ to $6.5$ times faster, in line with the cost model of Section~\ref{ssec:complexity}.

\subsection{Error attribution}
\label{ssec:error_attribution}

Table~\ref{tab:floors_v6} reports the representation and prediction terms of the error bound \eqref{eq:error_decomp} on held-out samples. Since the bound is a triangle inequality, its two terms do not split the error into additive parts. The prediction terms of the two methods differ by less than $0.1$ percentage points on every benchmark, and the remaining M1 minus M2 differences are consistent with the larger representation term of M1. Both methods use the same parameter neighborhoods, the same kernel weights and the same local Grassmann interpolation, and they differ in the field representation and in the coefficients they predict.

\begin{table}[!htb]
\centering
\caption{Representation and prediction terms of the bound \eqref{eq:error_decomp}, in percent, as means over ten resampled splits, with the original-split errors in Table~\ref{tab:results_summary} and values below $0.005$ shown as $0.00$.}
\label{tab:floors_v6}
\small
\begin{tabular}{lrrrrrr}
\toprule & \multicolumn{2}{c}{Representation} & \multicolumn{2}{c}{Prediction} & \multicolumn{2}{c}{Test error} \\
\cmidrule(lr){2-3}\cmidrule(lr){4-5}\cmidrule(lr){6-7}
Case & M1 & M2 & M1 & M2 & M1 & M2 \\
\midrule
I & 0.09 & 0.00 & 1.18 & 1.18 & 1.18 & 1.18 \\
II & 2.14 & 1.54 & 22.27 & 22.26 & 22.40 & 22.33 \\
III & 0.89 & 0.00 & 7.82 & 7.82 & 7.95 & 7.82 \\
IV & 0.48 & 0.00 & 1.75 & 1.76 & 1.88 & 1.76 \\
V & 1.82 & 0.24 & 6.72 & 6.81 & 7.04 & 6.82 \\
\bottomrule
\end{tabular}
\end{table}

Relative to the error level, the M1 minus M2 differences are largest where the representation term is a large share of the total. On Cases~I to~III the representation term is at most about an eighth of the total error, and the two methods are very close despite their different compression. On Cases~IV and~V the M1 representation term is about a quarter as large as the total error, and those are the two cases with the largest M1 minus M2 differences relative to the error level.

\section{Conclusion}
\label{sec:conclusion}

This work develops a tensor-based surrogate for repeated full-field prediction in parametric PDE problems, with the spatial directions and time kept as separate Tucker modes. The factor subspaces are predicted on a product Grassmann manifold, and Procrustes alignment places the Tucker cores in a common coordinate system before their parameter dependence is modeled by a low-rank aligned-core model. The predicted factors and core therefore combine directly in the reconstruction.

Across five PDE benchmarks the product-Grassmann surrogate is close in accuracy to a single-subspace Grassmannian surrogate with the same local interpolation, while reducing the interpolation size and the stored model by more than an order of magnitude and the cached query time by a factor of at least three. On held-out samples the two methods have almost the same prediction term under this common interpolation, and the final differences are larger where the Tucker representation error is a noticeable share of the total. The three-dimensional grid study shows the computational advantage growing with the spatial resolution. The full-field reconstruction limits the realized end-to-end speedup.

The current implementation uses fixed Tucker ranks across the parameter domain, tensor modes tied to the grid axes and one reference tangent space for the Grassmann interpolation, which requires the training and query subspaces to remain close to that reference. The main numerical study covers scalar fields on regular grids with two or three parameters. A supplementary study on Case~III raises the parameter dimension to six, where both methods lose accuracy and the M1 minus M2 upper confidence limit exceeds the one-point tolerance. Widening the transverse design box of Case~V to twice its reported width opens the gap to $5.30$ points once training subspaces at opposite ends of the box turn through most of a right angle relative to the reference.

A natural extension is to use several local reference spaces when the subspaces vary too far for a single reference. A second direction is validation on larger three-dimensional finite-volume or finite-element simulations and on measured spatiotemporal fields.

\if1\blind
{
\section*{Declaration of competing interest}
The authors report no competing interests.
} \fi

\if0\blind
{
\section*{CRediT authorship contribution statement}
\textbf{Melika Baghi:} Conceptualization, Methodology, Software, Validation, Investigation, Writing -- original draft. \textbf{Xiao Liu:} Conceptualization, Supervision, Writing -- review \& editing. \textbf{Kamran Paynabar:} Methodology, Supervision, Writing -- review \& editing.

\section*{Declaration of competing interest}
The authors declare no competing interests.

\section*{Declaration of generative AI and AI-assisted technologies in the writing process}
During the preparation of this work the authors used ChatGPT (OpenAI) and Claude (Anthropic) to improve the language and organization of the manuscript. After using these tools, the authors reviewed and edited the content as needed and take full responsibility for the content of the publication.
} \fi

\section*{Funding}
This research did not receive any specific grant from funding agencies in the public, commercial, or not-for-profit sectors.
\section*{Data and code availability}
During review, the implementation, the experiment scripts, the stored run outputs and the build file, whose three verification scripts recompute every sourced number of both documents from those outputs, are available from the corresponding author on request. Upon acceptance the repository will be made public under a persistent identifier, with the generator scripts that regenerate Cases~I, III and~V under their fixed seeds, the stored fields of Case~II, the case identifiers and preprocessing script of the external data, the software environment with its fixed random seeds, and the files behind every reported table and figure.

\clearpage
\appendix
\setcounter{figure}{0}\setcounter{table}{0}
\makeatletter
\@addtoreset{figure}{section}\@addtoreset{table}{section}
\@addtoreset{theorem}{section}\@addtoreset{proposition}{section}\@addtoreset{lemma}{section}
\@addtoreset{corollary}{section}\@addtoreset{assumption}{section}\@addtoreset{remark}{section}
\makeatother
\renewcommand{\thetheorem}{\Alph{section}.\arabic{theorem}}
\renewcommand{\theproposition}{\Alph{section}.\arabic{proposition}}
\renewcommand{\thelemma}{\Alph{section}.\arabic{lemma}}
\renewcommand{\thecorollary}{\Alph{section}.\arabic{corollary}}
\renewcommand{\theassumption}{\Alph{section}.\arabic{assumption}}
\renewcommand{\theremark}{\Alph{section}.\arabic{remark}}

\section[Technical Details for the Product-Grassmann Tucker Surrogate]{Technical Details for the\\ Product-Grassmann Tucker Surrogate}
\label{app:technical_details}

This appendix provides the bounds, proofs and algorithms behind the Methodology section of the main paper. One of the proofs establishes the coordinate compatibility of the predicted factors and the predicted aligned core, which lets them combine without online realignment. Everything is written for the third-order case.

\subsection{Representation error}
\label{app:repr_error}

For the matricized snapshot matrix $M_i$ the best rank-$r_s$ spatial subspace is spanned by its leading left singular vectors, and the squared residual is the discarded singular energy $\sum_{j>r_s}\sigma_j(M_i)^2$~\citep{Eckart1936}. For the tensor, the best Tucker approximation at ranks $(r_X,r_Y,r_T)$, with squared Frobenius error $E_{\mathrm{Tucker}}(\mu)$, depends on the factors only through the subspaces $[U_\bullet]$. The truncated higher-order SVD~\citep{DeLathauwer2000a,Kolda2009} satisfies the mode-wise bound
\begin{equation}
\label{eq:hosvd_bound}
\bigl\|\mathcal{X}(\mu)-\hat{\mathcal{X}}_{\mathrm{HOSVD}}(\mu)\bigr\|_F^2
\le
\sum_{\bullet\in\{X,Y,T\}}\;\sum_{j>r_\bullet}
\bigl(\sigma_j^{(\bullet)}(\mu)\bigr)^2 ,
\end{equation}
with $\sigma_j^{(\bullet)}(\mu)$ the $j$th singular value of the mode-$\bullet$ matricization. The truncated HOSVD is quasi-optimal, with a Frobenius error within the factor $\sqrt{3}$ of that of the best Tucker approximation at the same ranks \citep{Hackbusch2012}, whose squared error is bracketed by the same tails,
\begin{equation}
\label{eq:tucker_floor_bracket}
\max_{\bullet}\sum_{j>r_\bullet}\bigl(\sigma_j^{(\bullet)}(\mu)\bigr)^{2}
\;\le\;
E_{\mathrm{Tucker}}(\mu)
\;\le\;
\sum_{\bullet}\sum_{j>r_\bullet}\bigl(\sigma_j^{(\bullet)}(\mu)\bigr)^{2},
\end{equation}
since truncating one mode alone already discards that mode's tail~\citep{Vannieuwenhoven2012}. With the third spatial axis of Case Study~V the factor $\sqrt{3}$ becomes $\sqrt{4}$. For an exact truncated HOSVD, the square root of \eqref{eq:hosvd_bound} divided by $\|\mathcal{X}(\mu)\|_F$ is the theoretical upper bound on the relative truncation error. The experiments report the realized truncation error of the implemented decomposition, whose factors on Cases~I--IV come from randomized mode-wise singular value decompositions.

\subsection{Product Grassmann geometry}
\label{app:product_grassmann}

The Grassmann manifold $\Gr(r,n)$ of $r$-dimensional subspaces of $\mathbb{R}^n$, its tangent space at a point, the principal angles $\theta_j=\arccos\sigma_j(U^\top\widetilde U)\in[0,\pi/2]$ between two subspaces and the induced distance $d_{\Gr}([U],[\widetilde U])=(\sum_j\theta_j^2)^{1/2}$ are standard \citep{Absil2008,Edelman1998}. Each training tensor gives the product point $\mathbf{U}_i=([U_{X,i}],[U_{Y,i}],[U_{T,i}])\in\mathcal{M}:=\Gr(r_X,X)\times\Gr(r_Y,Y)\times\Gr(r_T,T)$. The product norm and distance add across modes, and the tangent space at the reference has dimension $r_X(X-r_X)+r_Y(Y-r_Y)+r_T(T-r_T)$ against $r_s(XY-r_s)$ for a single spatial subspace on $\Gr(r_s,XY)$. The reference is a mean of the training product points,
\begin{equation}
\label{app:eq_refmean}
\mathbf{U}_0
\in
\arg\min_{\mathbf{V}\in\mathcal{M}}
\sum_{i=1}^{N}
d_{\mathcal{M}}^2(\mathbf{V},\mathbf{U}_i),
\end{equation}
the Riemannian center of mass~\citep{karcher1977}, computed by the usual fixed-point iteration and unique when the training points lie in a small enough geodesic ball~\citep{afsari2011}. Principal angles are reported throughout as multiples of $\pi/2$, the value at which the logarithm ceases to be single-valued.

\subsection{Grassmann logarithm and exponential maps}
\label{app:log_exp}

Let $[U_0]\in\Gr(r,n)$ be the reference subspace in one mode and $[U]$ a subspace in a normal neighborhood of it, handled through its horizontal representative $Z$ with $U_0^\top Z=0$. With the thin singular value decomposition
\begin{equation}
\label{app:eq_log_svd}
(I-U_0U_0^\top)\,U\,(U_0^\top U)^{-1}
=
\widetilde U\,\Sigma\,\widetilde V^\top ,
\end{equation}
which requires the overlap $U_0^\top U$ to be invertible,
\begin{equation}
\label{app:eq_log_map}
\log_{[U_0]}([U])
=
\widetilde U\arctan(\Sigma)\,\widetilde V^\top,
\end{equation}
and, for a horizontal tangent vector with singular value decomposition $Z=W\Theta\widetilde V^\top$,
\begin{equation}
\label{app:eq_exp_map}
\exp_{[U_0]}(Z)
=
U_0\widetilde V\cos(\Theta)\,\widetilde V^\top
+
W\sin(\Theta)\,\widetilde V^\top.
\end{equation}
Both maps are standard~\citep{Absil2008,Edelman1998,Bendokat2020}. The product-manifold maps apply these formulas to the $X$, $Y$ and $T$ factors independently, and the singular values in \eqref{app:eq_log_svd} are the tangents of the principal angles between $[U]$ and $[U_0]$.

\begin{assumption}[Single reference neighborhood]
\label{app:assump_chart}
The training product points and the query product point lie in a common local
neighborhood of $\mathbf{U}_0$ on which the mode-wise Grassmann logarithm maps
are uniquely defined (equivalently, every principal angle to the reference
is strictly below $\pi/2$) and the interpolated tangent coordinates
remain within this neighborhood.
\end{assumption}

The implementation uses $\theta_{\mathrm{safe}}=0.9\,(\pi/2)$ as an offline diagnostic margin and as the online clipping threshold. In overlap terms the margin reads $\sigma_{\min}(U_{\bullet,0}^\top U_{\bullet,i})\ge\cos\theta_{\mathrm{safe}}$. Offline, the largest training principal angle of every benchmark is measured against it, and training samples past the margin are admitted as long as their angles are below $\pi/2$, where the logarithm is single-valued. Case~IV exceeds the margin, with a largest training angle of $0.97\,(\pi/2)$, while remaining inside that region. Online, the largest singular value of each predicted representative $\hat Z_\bullet(\mu)$ is checked before the exponential map, and the representative is rescaled to $\theta_{\mathrm{safe}}$ when that value exceeds the threshold.

\subsection{Procrustes alignment and coordinate compatibility}
\label{app:procrustes}

With $U_{\bullet,0}$ the reference basis in mode $\bullet\in\{X,Y,T\}$ and the singular value decomposition of the overlap
\begin{equation}
\label{app:eq_overlap_svd}
U_{\bullet,0}^\top U_{\bullet,i}
=
L_{\bullet,i}\Sigma_{\bullet,i}R_{\bullet,i}^\top ,
\end{equation}
the Procrustes rotation of the main text is
\begin{equation}
\label{app:eq_procrustes_rotation}
Q_{\bullet,i}=R_{\bullet,i}L_{\bullet,i}^\top.
\end{equation}
The aligned factor is $\bar U_{\bullet,i}=U_{\bullet,i}Q_{\bullet,i}$ and the aligned core is $\bar{\mathcal{G}}_i=\mathcal{G}_i\times_1 Q_{X,i}^\top\times_2 Q_{Y,i}^\top\times_3 Q_{T,i}^\top$.

\begin{assumption}[Nonsingular overlaps]
\label{app:assump_overlap}
For every mode $\bullet\in\{X,Y,T\}$ and every training sample $i$, the
overlap matrix $U_{\bullet,0}^\top U_{\bullet,i}$ is nonsingular. Equivalently, no principal angle between the sample and reference subspaces
reaches $\pi/2$.
\end{assumption}

Under Assumption~\ref{app:assump_overlap} the rotation is unique and the alignment is well defined.

Rotating a factor from $U$ to $UQ$ keeps the subspace unchanged and forces the core to absorb the inverse rotation,
\begin{equation}
\label{app:eq_counter_rotation}
\mathcal{G}\times_n U=\bigl(\mathcal{G}\times_n Q^\top\bigr)\times_n (UQ).
\end{equation}
A predicted factor in an arbitrary orientation therefore does not combine with the predicted aligned core. On the training side, Procrustes fixes each aligned factor as the representative of its subspace whose overlap with the reference is symmetric positive definite, the cores are counter-rotated to match, and the aligned-core model is trained in those coordinates. Alignment is an exact change of coordinates: the factor rotation and the core counter-rotation cancel in the reconstruction by \eqref{app:eq_counter_rotation}, and the step contributes no approximation error.

\begin{proposition}[Coordinate compatibility]
\label{app:prop_coordinate}
Under Assumptions~\ref{app:assump_chart} and~\ref{app:assump_overlap}, the
Grassmann exponential map \eqref{app:eq_exp_map} and the Procrustes alignment
\eqref{app:eq_procrustes_rotation} return, in each mode, the unique orthonormal
representative of their argument subspace whose overlap with the reference basis
$U_{\bullet,0}$ is symmetric and positive definite. The exponential-predicted
query factors and the Procrustes-aligned training factors are therefore
expressed in one common reference frame. The predicted core then combines
with the predicted factors in the Tucker reconstruction without any further
rotation.
\end{proposition}

\begin{proof}
Let $\hat Z=W\Theta\widetilde V^\top$ be the thin SVD of the predicted tangent vector in mode $\bullet$, with $U_{\bullet,0}^\top W\sin\Theta=0$ and $\Theta=\mathrm{diag}(\theta_1,\ldots,\theta_{r_\bullet})$, $\theta_j\ge0$. Multiplying \eqref{app:eq_exp_map} on the left by $U_{\bullet,0}^\top$ gives
\begin{equation}
\label{app:eq_exp_overlap}
U_{\bullet,0}^\top\hat{U}_\bullet
=
\widetilde{V}\cos\Theta\,\widetilde{V}^\top,
\end{equation}
symmetric, and positive definite since every $\theta_j<\pi/2$ under Assumption~\ref{app:assump_chart}. On the Procrustes side, \eqref{app:eq_overlap_svd} and \eqref{app:eq_procrustes_rotation} give
\begin{equation}
\label{app:eq_procrustes_overlap}
U_{\bullet,0}^\top\bar{U}_{\bullet,i}
=
L_{\bullet,i}\Sigma_{\bullet,i}R_{\bullet,i}^\top R_{\bullet,i}L_{\bullet,i}^\top
=
L_{\bullet,i}\Sigma_{\bullet,i}L_{\bullet,i}^\top,
\end{equation}
symmetric, and positive definite since $\Sigma_{\bullet,i}\succ0$ under Assumption~\ref{app:assump_overlap}. For a subspace $[U]$ with nonsingular overlap $U_{\bullet,0}^\top U=L\Sigma R^\top$, the overlap $U_{\bullet,0}^\top UQ$ is symmetric positive definite for exactly one $Q\in O(r_\bullet)$, namely $Q=RL^\top$. Both \eqref{app:eq_exp_overlap} and \eqref{app:eq_procrustes_overlap} have the form $A\Delta A^\top$ with $A$ orthogonal and $\Delta\succ0$ diagonal. Each construction therefore returns that unique representative of its own subspace, anchored to the same $U_{\bullet,0}$.
\end{proof}

\subsection{Aligned-core bound}
\label{app:core_of_cores}

With $G\in\mathbb{R}^{(r_Xr_Yr_T)\times N}$ the matrix of vectorized aligned-core residuals about their mean, $G\approx U_p\Sigma_pV_p^\top$ its rank-$r_p$ truncated SVD and $w(\mu)$ the kernel weights of a query, the aligned-core model of the main text predicts $\bar{\mathcal{G}}(\mu)=\bar{\mathcal{G}}_{\mathrm{mean}}+\mathcal{G}_{\mathrm{basis}}\bar{\times}_4 V_p^\top w(\mu)$ with $\mathcal{G}_{\mathrm{basis}}=\mathrm{reshape}(U_p\Sigma_p,[r_X,r_Y,r_T,r_p])$, and the full field is
\begin{equation}
\label{app:eq_reconstruction}
\hat{\mathcal{X}}(\mu)
=
\bar{\mathcal{G}}(\mu)
\times_1 \bar U_X(\mu)
\times_2 \bar U_Y(\mu)
\times_3 \bar U_T(\mu),
\qquad
\bar U_\bullet(\mu):=\hat U_\bullet(\mu),
\end{equation}
where Proposition~\ref{app:prop_coordinate} justifies the identification on the right.

Because the factors are orthonormal, truncating the core model at rank $r_p$ changes the field by exactly the norm of the discarded core residual. With $\hat{\mathcal{X}}_{\mathrm{full}}(\mu)$ the field reconstructed from the untruncated residual model and $\hat{\mathcal{X}}_{r_p}(\mu)$ the field from the rank-$r_p$ model,
\begin{equation}
\label{app:eq_coc_propagation}
\begin{split}
\bigl\|\hat{\mathcal{X}}_{\mathrm{full}}(\mu)-\hat{\mathcal{X}}_{r_p}(\mu)\bigr\|_F
&=\bigl\|(I-U_pU_p^{\top})G\,w(\mu)\bigr\|_2\\
&\le\sigma_{r_p+1}(G)\,\|w(\mu)\|_2
\le\sigma_{r_p+1}(G),
\end{split}
\end{equation}
the last step because the weights are nonnegative and sum to one. The aligned-core truncation error is therefore bounded by a single discarded singular value. At the selected $r_p$ the bound is $0.34$, $4.51$, $0.83$, $1.62$ and $1.34$ percent of the mean training field norm for Cases~I--V. It is a worst-case bound over an arbitrary weight vector.

\subsection{Algorithmic summary}
\label{app:algorithm}

\begin{algorithm}[!htb]
\caption{Offline training of the product-Grassmann Tucker surrogate}
\label{app:alg_offline}
\begin{algorithmic}[1]
\Require $\{(\mu_i,\mathcal{X}_i)\}_{i=1}^{N}$, candidate ladders for $(r_X,r_Y,r_T)$, $r_p$, $k$
\State $(r_X,r_Y,r_T),\,r_p,\,k \gets$ leave-one-out selection \Comment{each score runs lines \ref{alg:hosvd_start}--\ref{alg:coc_end}}
\For{$i=1,\dots,N$}\label{alg:hosvd_start}
  \State $(\mathcal{G}_i;\,U_{X,i},U_{Y,i},U_{T,i}) \gets \mathrm{HOSVD}_{(r_X,r_Y,r_T)}(\mathcal{X}_i)$\label{alg:hosvd}
\EndFor
\For{$\bullet\in\{X,Y,T\}$}
  \State $U_{\bullet,0} \gets \arg\min_{V}\sum_{i} d_{\Gr}^{2}([V],[U_{\bullet,i}])$ \Comment{fixed-point iteration, Eq.~\eqref{app:eq_refmean}}
  \For{$i=1,\dots,N$}
    \State $L\Sigma R^\top \gets \mathrm{svd}(U_{\bullet,0}^\top U_{\bullet,i})$,\quad $Q_{\bullet,i}\gets RL^\top$ \Comment{Eq.~\eqref{app:eq_procrustes_rotation}}
    \State $\bar U_{\bullet,i}\gets U_{\bullet,i}Q_{\bullet,i}$ \Comment{aligned factor}
    \State $Z_{\bullet,i}\gets \log_{[U_{\bullet,0}]}([U_{\bullet,i}])$ \Comment{Eq.~\eqref{app:eq_log_map}}
  \EndFor
\EndFor
\For{$i=1,\dots,N$}
  \State $\bar{\mathcal{G}}_i \gets \mathcal{G}_i\times_1 Q_{X,i}^\top\times_2 Q_{Y,i}^\top\times_3 Q_{T,i}^\top$ \Comment{aligned core}
\EndFor
\State $\bar{\mathcal{G}}_{\mathrm{mean}}\gets \frac{1}{N}\sum_{i}\bar{\mathcal{G}}_i$,\quad $G(:,i)\gets \mathrm{vec}(\bar{\mathcal{G}}_i-\bar{\mathcal{G}}_{\mathrm{mean}})$
\State $U_p\Sigma_pV_p^\top \gets \mathrm{svd}_{r_p}(G)$ \Comment{aligned-core model}
\State $\mathcal{G}_{\mathrm{basis}}\gets \mathrm{reshape}(U_p\Sigma_p,[r_X,r_Y,r_T,r_p])$\label{alg:coc_end}
\Ensure $\mathbf{U}_0=(U_{X,0},U_{Y,0},U_{T,0})$,\; $\{Z_{\bullet,i}\}$,\; $\bar{\mathcal{G}}_{\mathrm{mean}}$,\; $\mathcal{G}_{\mathrm{basis}}$,\; $V_p$
\end{algorithmic}
\end{algorithm}

\begin{algorithm}[!ht]
\caption{Online prediction at a query parameter $\mu$}
\label{app:alg_online}
\begin{algorithmic}[1]
\Require $\mu$, the output of Algorithm~\ref{app:alg_offline}, $k$, $\theta_{\mathrm{safe}}=0.9\,(\pi/2)$
\State $w\gets$ Gaussian kernel weights on the $k$ nearest training parameters, $\sum_i w_i=1$
\For{$\bullet\in\{X,Y,T\}$}
  \State $\hat Z_\bullet\gets \sum_{i:\,w_i>0} w_i Z_{\bullet,i}$
  \State \textbf{if} $\sigma_{\max}(\hat Z_\bullet)>\theta_{\mathrm{safe}}$ \textbf{then} $\hat Z_\bullet\gets \bigl(\theta_{\mathrm{safe}}/\sigma_{\max}(\hat Z_\bullet)\bigr)\hat Z_\bullet$ \Comment{safety margin}
  \State $\hat U_\bullet\gets \exp_{[U_{\bullet,0}]}(\hat Z_\bullet)$ \Comment{Eq.~\eqref{app:eq_exp_map}}
  \State $\bar U_\bullet(\mu)\gets \hat U_\bullet$ \Comment{no alignment step, Prop.~\ref{app:prop_coordinate}}
\EndFor
\State $a\gets V_p^\top w$,\quad $\bar{\mathcal{G}}(\mu)\gets \bar{\mathcal{G}}_{\mathrm{mean}}+\mathcal{G}_{\mathrm{basis}}\,\bar{\times}_4\,a$ \Comment{aligned-core model}
\State $\hat{\mathcal{X}}(\mu)\gets \bar{\mathcal{G}}(\mu)\times_1\bar U_X(\mu)\times_2\bar U_Y(\mu)\times_3\bar U_T(\mu)$ \Comment{Eq.~\eqref{app:eq_reconstruction}}
\Ensure $\hat{\mathcal{X}}(\mu)$
\end{algorithmic}
\end{algorithm}

In the experiments, the decomposition of line~\ref{alg:hosvd} of Algorithm~\ref{app:alg_offline} uses randomized mode-wise SVDs on Cases~I--IV and full mode-wise SVDs on Case~V. The implementation re-applies the Procrustes rotation to each predicted factor as a floating-point safeguard. Proposition~\ref{app:prop_coordinate} shows that this rotation is the identity up to numerical roundoff. A field with a third spatial axis adds one Grassmann factor and one core index. Every mode-wise step then runs once more, and the reconstruction becomes $\bar{\mathcal{G}}(\mu)\times_1\bar U_X(\mu)\times_2\bar U_Y(\mu)\times_3\bar U_Z(\mu)\times_4\bar U_T(\mu)$. Case Study~V uses this form.

\section{Supplemental Numerical Model Details}
\label{app:numerical_details}

This section records the constants, designs, data provenance and implementation details behind the numerical experiments. The governing equations are in the main text.

\subsection{Case Study I: anisotropic diffusion-reaction}
\label{app:case1_details}

The parameter vector is $\mu=(D_x,D_y,\lambda)\in[0.003,0.008]\times[0.002,0.005]\times[0.40,0.90]$. The $40$ samples are a Latin hypercube over this box under a fixed seed; $30$ are used for training and $10$ are held out for testing, chosen among the interior points, at least $0.15$ of every parameter range from the boundary. A fixed rule kept in the repository picks the test points so that they are spread out and each has training neighbors around it. The source is a fixed separable Gaussian in space times a fixed periodic factor in time, with the constants in the repository. The fields are computed on the $256\times256$ grid by a Fourier spectral method on the periodic square with an integrating-factor step on the linear operator and the forcing evaluated at the step midpoint, and stored at $100$ uniformly spaced times over $[0,1]$.

\subsection{Case Study II: double-slit wave interference}
\label{app:case2_details}

\begin{sloppypar}
The $36$ runs form the full $6\times6$ grid of amplitude $\mu_1\in\{60,72,84,96,108,120\}$ against frequency parameter $\mu_2\in\{3,3.5,4,4.5,5,5.5\}$; $30$ are used for training and $6$ interior grid points, chosen by a fixed interior-and-separation rule kept in the repository, are held out for testing. The archive holds the stored fields, the time vector and the parameter pair of each run. It does not record the integration scheme that produced them, the slit positions or the treatment of the other boundaries. The frequency parameter $\mu_2$ sets a cyclic
frequency of $\mu_2/2$. The initial condition is quiescent. The first stored frame, at $t=0.002$, already carries the disturbance the boundary has launched.
\end{sloppypar}

\subsection{Case Study III: multimodal forced heat equation}
\label{app:case3_details}

The parameter vector is $\mu=(\mu_1,\mu_2)\in[0.5,1.5]^2$. Of $40$ uniform draws over this square under a fixed seed, the first $30$ train the surrogates and the last $10$ test them. The initial field is a sum of eight sine modes, $u(x,y,0;\mu)=\sum_{k=1}^{8}a_k(\mu)\sin(m_k^x\pi x)\sin(m_k^y\pi y)$, and the forcing a sum of nine Gaussian-in-time pulses on sine modes, $S(x,y,t;\mu)=\sum_{p=1}^{9}A_p(\mu)\exp(-(t-t_p)^2/(2\sigma_p^2))\sin(m_p^x\pi x)\sin(m_p^y\pi y)$, with every mode index in $\{1,2,3\}$, pulse centers $t_p$ between $0.05$ and $0.38$ and widths $\sigma_p$ of $0.020$ to $0.030$. The spatial content of every run is therefore spanned by three sine functions per axis. The amplitude maps $a_k(\mu)$ and $A_p(\mu)$, sinusoids of $\mu_1$ scaled on some terms by $\mu_2$, and the diffusivity $D(\mu_2)$ of the governing equation depend on the parameter. Their constants are those of the generator in the repository. The fields are integrated on the $128\times128$ grid by a Peaceman--Rachford alternating-direction scheme and stored at $101$ times.

\subsection{Case Study IV: tube flow (CFDBench)}
\label{app:case5_details}

The fields come from the \texttt{tube} problem of CFDBench~\citep{luo2023cfdbench,cfdbench2023data}, whose authors document the solver and the data. Of its three subsets the one varying the fluid properties is used, since it keeps a single domain and varies two parameters. The subset provides a $10\times10$ grid over the density index and the viscosity, both used as distributed. The stored runs differ in length. A common length of $15$ frames is fixed and the two of the $100$ cases shorter than that are dropped, leaving $98$; $78$ are used for training and $20$, drawn at random under a fixed seed, are held out for testing, $8$ of them on the boundary of the grid. The two stored velocity components are reduced to their magnitude. The archive records no numeric time step, and the frame index is the common temporal coordinate. Repeating the study at $18$ frames, where the rule selects ranks $(20,20,18)$, gives $1.35$ percent for M1 against $1.25$ for M2 on the original split, and the choice of length does not decide the comparison. Ranks $(20,20,15)$ are selected by the rule of the main text, with the temporal mode retained in full. The reference is initialization dependent here, with a spread of $0.82\,(\pi/2)$ across seeds in the first spatial mode, consistent with a maximum training angle of $0.97\,(\pi/2)$. The test error changes by less than $0.001$ percentage points over five references, because every downstream quantity is recomputed against whichever reference the iteration returns and Proposition~\ref{app:prop_coordinate} holds for each.

\subsection{Case Study V: three-dimensional advection-diffusion}
\label{app:case3d_details}

The initial condition superposes two structures of different width,
\begin{equation}
\begin{split}
u(\vect{x},0;\mu)=\;&
\exp\!\left[-\frac{d(x,x_0)^2+d(y,\mu_y)^2+d(z,\mu_z)^2}{2\sigma_n^2}\right]\\[2pt]
&+\;\alpha\,
\exp\!\left[-\frac{d(x,x_b)^2+d(y,y_b)^2+d(z,z_b)^2}{2\sigma_w^2}\right],
\end{split}
\label{eq:case3d_ic}
\end{equation}
with $\vect{x}=(x,y,z)$ the spatial coordinate, $\sigma_n=0.15$ and $\sigma_w=0.28$ the widths of the narrow and the wide structure, $x_0=0.14$, $(x_b,y_b,z_b)=(0.62,0.40,0.62)$, and $d(a,b)=((a-b+\tfrac12)\bmod 1)-\tfrac12$ the signed periodic distance. The parameter vector is $\mu=(\mu_y,\mu_z,\alpha)\in[0.40,0.60]^2\times[0.15,1.00]$. The $50$ samples are a Latin hypercube over this box under a fixed seed; $40$ are used for training and $10$, drawn at random from the interior points at least a tenth of each side's range from its edges, are held out for testing, and no prediction extrapolates. The propagator is applied in Fourier space with NumPy's FFT on a periodic $32^3$ grid. It evaluates the spatially discretized evolution directly at the stored times. No time-stepping error enters the comparison, while the fields inherit the spatial discretization of the $32^3$ Fourier grid. Each sample is a $32\times32\times32\times24$ tensor of $786{,}432$ entries, regenerated on demand.

\subsection{Implementation and hyperparameter selection}
\label{app:surrogate_impl}

Across all cases the interpolation weights use a Gaussian kernel over the $k$ nearest training parameters in a standardized parameter space, each coordinate centered and scaled to unit variance with the training statistics and neighbors ranked by Euclidean distance. The bandwidth $h$ is set per query as the larger of the standard deviation and half the mean of the $k$ neighbor distances, and $w_i\propto\exp(-d_i^2/(2h^2))$, normalized to sum to one over the selected neighbors. M1's Tucker factors are formed from randomized mode-wise singular value decompositions~\citep{halko2011randomized} on Cases~I--IV, through scikit-learn~1.8.0 with its default oversampling and power-iteration settings and a fixed random state of $42$, and from full mode-wise SVDs on Case~V. The spatial basis of M2 comes from a full SVD in every case, in double precision. The tuples the leave-one-out rule of the main text selects carry mean training truncation errors of $0.10$, $2.16$, $0.88$, $0.47$ and $1.88$ percent for Cases~I--V. On Case~V the $r_s$ ladder is scored on ten evenly spaced training folds, since every fold refits the single-subspace surrogate on $\Gr(r_s,32{,}768)$, and every other ladder of every case is scored on all training folds. The accuracy results of Cases~I--IV use a fixed random seed of $42$ under NumPy~2.4.2, and Case~V its own fixed seed of $7$ for the design and the original split. Table~\ref{tab:app_config} lists the neighborhood size and core rank each benchmark selects on training data, beside the largest principal angle between a training factor and the reference and the smallest singular value of a training-to-reference overlap. The ranks and $r_s$ are those of the cross-case table of the main text.
\begin{table}[!htb]
\centering
\caption{Selected neighborhood size $k$ and core rank $r_p$, largest training angle in units of $\pi/2$ and smallest overlap singular value, on the original split.}
\label{tab:app_config}
\footnotesize
\begin{tabular}{lcccc}
\toprule
Case & $k$ & $r_p$ & Max angle & Overlap $\sigma_{\min}$ \\
\midrule
I & $5$ & $3$ & $0.27$ & $0.91$ \\
II & $3$ & $5$ & $0.46$ & $0.75$ \\
III & $3$ & $5$ & $0.84$ & $0.25$ \\
IV & $4$ & $10$ & $0.97$ & $0.047$ \\
V & $5$ & $4$ & $0.61$ & $0.569$ \\
\bottomrule
\end{tabular}
\end{table}

The tenth-of-a-point band on the rank ladders avoids raising the rank for validation improvements smaller than that, since rank is expensive online and the leave-one-out score flattens toward the top of a ladder. The rule follows the tensor-on-tensor regression of \citet{reisi2021mtot}, which cross-validates its ranks over a grid, with a short fixed ladder as the grid. An earlier form of the rule applied the hundredth-of-a-point band to every ladder and selected higher rungs on the wave and heat cases, $(13,19,19)$ with $r_s=45$ and $(3,3,7)$. The tenth-of-a-point band of the reported analysis was adopted after a comparison of those configurations on the original benchmark splits. The reported configurations of the original splits therefore depend on that comparison, and the nested study of Table~\ref{tab:app_nested} does not inherit it. On each of the ten resampled splits every model choice, the Tucker ranks, $r_p$, $k$ and $r_s$, is re-selected inside that split's training set by the leave-one-out rule, under the tenth-of-a-point band and, from the same leave-one-out scores, under the hundredth-of-a-point band. The rank ladder is scored at the provisional $(k,r_p)$ of the reported protocol, and on Case~V the $r_s$ ladder on ten evenly spaced training folds. No test sample of the split enters the selection of the ranks, $r_p$, $k$ or $r_s$. Under either band every upper limit stays below the one-point tolerance. The two bands select the same ranks on Cases~I and~IV and on nine of the ten splits of Case~V, while on the wave and heat cases the narrower band raises the ranks on every split and moves the mean gap by at most $0.05$ points.

\begin{table}[!htb]
\centering
\caption{Nested selection on the ten resampled splits under the $0.1$- and $0.01$-point bands: the rank configurations selected and the mean M1 minus M2 gap in points with its paired $95\%$ interval.}
\label{tab:app_nested}
\footnotesize
\setlength{\tabcolsep}{4pt}
\begin{tabular}{lllll}
\toprule
& \multicolumn{2}{c}{Band $0.1$} & \multicolumn{2}{c}{Band $0.01$} \\
\cmidrule(lr){2-3}\cmidrule(lr){4-5}
Case & M1 ranks & Gap $[95\%]$ & M1 ranks & Gap $[95\%]$ \\
\midrule
I & $(2,4,4)$ & $-0.00\ [-0.01,+0.00]$ & $(2,4,4)$ & $-0.00\ [-0.01,+0.00]$ \\
II & $(11,16,16)$ & $+0.36\ [+0.09,+0.63]$ & $(13,19,19)$ & $+0.31\ [+0.08,+0.55]$ \\
III & $(3,3,5)$ & $+0.05\ [+0.03,+0.07]$ & $(3,3,7)$ & $+0.00\ [-0.00,+0.01]$ \\
IV & $(20,20,15)$ & $+0.13\ [+0.09,+0.17]$ & $(20,20,15)$ & $+0.13\ [+0.09,+0.17]$ \\
V & $(5,2,2,5)$ & $+0.22\ [+0.16,+0.27]$ & $(5,2,2,5), (8,4,4,8)$ & $+0.19\ [+0.13,+0.25]$ \\
\bottomrule
\end{tabular}
\end{table}
Under the reported band the splits select $r_s=20$ or $45$ on Case~II, and the median interpolation-size ratio there is $29.0\times$ with range $[12.9,\,29.0]$.

Before the leave-one-out rule, a truncation-floor rule set the M1 ranks as the smallest tuple whose mean HOSVD truncation error on the training fields falls below three percent, while the single-subspace rank was still chosen by validation error. That asymmetry favors whichever method is tuned for accuracy over compression. Both rules are therefore reported. Table~\ref{tab:rank_rule_v6} sets them side by side, with the gap as mean M1 minus M2 over ten redraws in points and both ratios at the single-subspace rank that leave-one-out selects. On Cases~I, III and~V the two rules select the same ranks. On Cases~II and~IV the leave-one-out rule gives up compression for accuracy, and the truncation-floor gap of Case~IV, $+1.00$ points, falls on the prespecified tolerance itself.

\begin{table}[!htb]
\centering
\caption{The two rank-selection rules side by side.}
\label{tab:rank_rule_v6}
\small
\begin{tabular}{lrrrcrrr}
\toprule
 & \multicolumn{3}{c}{Truncation floor} & & \multicolumn{3}{c}{Leave-one-out} \\
\cmidrule{2-4}\cmidrule{6-8}
Case & ranks & gap & ratio & & ranks & gap & ratio \\
\midrule
I & (2,4,4) & $-0.00$ & $169.3\times$ & & (2,4,4) & $-0.00$ & $169.3\times$ \\
II & (9,13,13) & $+0.51$ & $15.8\times$ & & (11,16,16) & $+0.07$ & $12.9\times$ \\
III & (3,3,4) & $+0.13$ & $125.8\times$ & & (3,3,4) & $+0.13$ & $125.8\times$ \\
IV & (12,12,12) & $+1.00$ & $35.8\times$ & & (20,20,15) & $+0.12$ & $22.1\times$ \\
V & (5,2,2,5) & $+0.22$ & $481.9\times$ & & (5,2,2,5) & $+0.22$ & $481.9\times$ \\
\bottomrule
\end{tabular}
\end{table}

The candidate ladders, five to nine rank tuples per case and the $r_s$, $r_p$ and $k$ ladders, were fixed before any run and are kept in the repository. The Case~III ladder also varies the temporal rank on its own at the exact spatial rank three, beside its shape-matched tuples. On Case~I every rank candidate falls inside the tie band (leave-one-out error $1.410$ to $1.414$ percent across the ladder), and the rule returns the smallest, $(2,4,4)$. On Case~II the test error is lowest at $k=4$ while the training-selected value is $k=3$, and the training choice is kept.

\subsection{Timing and rank sensitivity}
\label{app:timing}

All runs execute on one workstation (Apple M2~Pro with 16~GB, NumPy~2.4.2 with Accelerate BLAS, one process, library-default threading). The cached query stores every training decomposition, tangent coordinate and coefficient matrix offline and sums over the active neighbors only, because the kernel weights are zero outside the $k$ nearest training parameters. The weights themselves, identical for the two methods, are computed before the timed region, and the timed query starts at the weighted tangent combination. These timings correspond to the cached many-query setting of the main text. Each test query is timed $50$ times after a warm-up pass, and the mean over all test queries is reported, with the Grassmann stage (weighted tangent combination, exponential maps, core or coefficient prediction) and the reconstruction stage (full-field assembly) timed separately. The M1 timings include the re-application of the Procrustes rotation to each predicted factor, the floating-point safeguard of \ref{app:algorithm}. M2 carries no equivalent step. Both methods use the safeguarded exponential map of Algorithm~\ref{app:alg_online}, with its clipping check and re-orthonormalization. For M2 the clip never activates on any reported split. Its safeguard still costs a thin singular value decomposition and a QR factorization of the basis, and that cost grows with the ambient dimension $XY$, or $XYZ$ on Case~V. All online and offline timings come from one session. Table~\ref{tab:app_offline} sets the measured interpolation ratio of each benchmark against the grid estimate $XY/(X+Y+T)$, or $XYZ/(X+Y+Z+T)$ on Case~V, and lists the offline training times with the break-even query count, the offline penalty of M1 divided by its cached per-query saving. The realized ratios fall within a small factor of the grid estimate. The order of magnitude therefore follows from the grid shape, and the rank selection corrects it. Training is faster for M1 on Cases~I and~IV and slower on Cases~II, III and~V. Rounded to two figures, the largest break-even count is $1{,}500$ queries, on Case~III, and the dashes mark the two cases where M1 is ahead from the first query.
\begin{table}[!htb]
\centering
\caption{Interpolation ratio against the grid estimate, offline training times on the original split and break-even query counts.}
\label{tab:app_offline}
\footnotesize
\setlength{\tabcolsep}{6pt}
\begin{tabular}{lrrrrr}
\toprule
& \multicolumn{2}{c}{Interpolation} & \multicolumn{2}{c}{Offline time (s)} & \\
\cmidrule(lr){2-3}\cmidrule(lr){4-5}
Case & Grid est. & Realized/est. & M1 & M2 & Break-even \\
\midrule
I & 107 & 1.6 & 7.7 & 32.1 & -- \\
II & 10 & 1.3 & 4.5 & 1.7 & 400 \\
III & 46 & 2.7 & 12.4 & 2.1 & 1,500 \\
IV & 29 & 0.8 & 1.8 & 2.7 & -- \\
V & 273 & 1.8 & 6.4 & 5.2 & 170 \\
\bottomrule
\end{tabular}
\end{table}

Because the interpolation-size ratio depends on M2's rank, which is selected independently, M2 is refitted at every rank of its ladder on the original split. Its error is flat above the selected rank on Cases~I, III and~V ($1.01$ percent from rank $5$ to rank $30$ on Case~I). A larger rank there inflates the ratio without helping M2. On the wave case the error falls by $0.05$ points from rank $20$ to rank $60$ while the ratio triples, and the tie band resolves that trade toward the smaller object. On Case~IV every candidate above $15$ fits the same model, because $15$ is that benchmark's matrix-rank ceiling.

\subsection{Resampling}
\label{app:resampling}

The resampling study of the main text redraws the held-out test set with seeds $1$ through $10$ and refits every method on each split, with ranks and hyperparameters fixed at the values above. Cases~I and~II redraw from the interior candidate pools of their original splits, Cases~III and~IV uniformly from the design, and Case~V from the interior of its box, a tenth of each side's range in from the edges. M1 and M2 are evaluated on identical splits and test samples, and the error differences therefore support a paired analysis. Table~\ref{tab:app_resampling} reports the mean M1 minus M2 gap over the ten splits with its paired $95\%$ $t$ interval and the relative margin, ten percent of the resampled M2 error. The resampled means themselves are the test-error columns of the main text's table of representation and prediction terms.

\begin{table}[!htb]
\centering
\caption{M1 minus M2 gap over ten resampled splits, in percentage points, as mean $\pm$ standard deviation, with its paired $95\%$ $t$ interval and the relative margin of ten percent of the M2 error.}
\label{tab:app_resampling}
\footnotesize
\begin{tabular}{lrrr}
\toprule
Case & M1$-$M2 & $95\%$ CI & Margin \\
\midrule
I   & $+0.00\pm0.01$ & $[-0.01,+0.00]$ & $0.12$ \\
II  & $+0.07\pm0.03$ & $[+0.04,+0.09]$ & $2.23$ \\
III & $+0.13\pm0.06$ & $[+0.09,+0.18]$ & $0.78$ \\
IV  & $+0.12\pm0.04$ & $[+0.09,+0.15]$ & $0.18$ \\
V   & $+0.22\pm0.09$ & $[+0.16,+0.28]$ & $0.68$ \\
\bottomrule
\end{tabular}
\end{table}

The observed M1 minus M2 difference counts as within tolerance when the upper limit of its paired $95\%$ interval stays below the one-point tolerance and below the relative margin, and every benchmark satisfies both criteria. On the wave, heat, external and three-dimensional benchmarks the whole interval lies above zero: M2 is more accurate there, by amounts more than an order of magnitude below the error levels. Because the test sets of two seeds overlap, a cluster bootstrap over simulations instead of splits ($10{,}000$ resamples, ten further redraws included) is the appropriate correction. Every cluster interval lies a factor of three or more below the tolerance, the widest being $[+0.08,+0.29]$ points on Case~V. Over the same splits the fixed core is $2.9$ to $5.0$ times worse than M1 and the unaligned core $2.2$ to $6.5$ times worse. The Procrustes alignment makes the aligned cores comparable enough to model. On the original split of each benchmark the fixed core gives $3.66$, $67.67$, $35.72$, $6.84$ and $19.12$ percent on Cases~I--V and the unaligned core $2.37$, $112.72$, $53.42$, $14.81$ and $17.37$ percent, against the M1 and M2 columns of the cross-case table of the main text.

On the wave case the gap to M2, $+0.07\pm0.03$ percentage points over the ten splits, reflects the temporal convention: M1 compresses time to rank $r_T$ while M2 keeps its full temporal coefficient matrix. Equipping M2 with a global rank-$r_T$ temporal basis costs it $0.063$ percentage points, $21.40$ against $21.34$ percent on the original split and $22.39\pm3.01$ against $22.33\pm3.03$ over the resampled splits. This brings it level with M1 and leaves every interpolation-size ratio unchanged, because the single-subspace storage is dominated by the spatial tangent matrices. The interpolation-size advantage of M1 therefore comes from the spatial tensorization, not from the temporal asymmetry of the comparison. On the heat case, where M1 truncates time to rank four against the $101$ coefficients of M2, the same variant costs M2 $0.217$ points, $6.63$ against $6.43$ percent on the original split and $8.03\pm1.64$ against $7.82\pm1.71$ over the resampled splits, more than the $+0.13$-point gap.

\subsection{Generalization and robustness checks}
\label{app:robustness}

To examine the sensitivity to the parameter dimension apart from the grid resolution, Case~III is repeated with the training count and the numerical discretization fixed while $d$ increases from $2$ to $6$. Mode $k$ takes its phase from $\mu_{1+(k \bmod (d-1))}$ and every mode takes its scale from $\mu_d$. This parameterization reproduces the main-text parameterization at $d=2$ and leaves the solver, grid, sample count, split and ranks untouched. As $d$ grows to $6$, both methods lose accuracy together, from $6.6$ and $6.4$ percent at $d=2$ to $19.7$ and $18.0$ percent at $d=6$. The mean training truncation error stays below $0.9$ percent, and the loss is prediction error. The gap holds at $+0.23$ points or below through $d=5$. It reaches $+0.98$ points at $d=6$ with an upper confidence bound of $1.20$, which exceeds the one-point tolerance.

A boundary split holds out the samples farthest from the design center in the standardized parameter space, where the surrogate predicts at the edge of the sampled domain with no surrounding neighbors. Errors rise for every method, by factors of $1.5$ to $3.8$ for M1, while M1 and M2 stay within $0.1$ percentage points on the two-dimensional benchmarks and $0.21$ on the three-dimensional one, where M1 is the lower of the two ($16.03$ against $16.24$ percent). Relative Gaussian noise at three percent of the field root-mean-square is added to every training tensor, and the surrogates are refitted on the ten resampled splits with two noise seeds per split. The noise moves the M1 error by less than $0.01$ points on Cases~I--III and~V and by $0.13$ points on the external benchmark, whereas M2 degrades on every case, the most on the heat case ($7.82$ to $9.34$ percent), with an excess over M1's degradation of up to $+1.52$ points whose paired interval excludes zero on every case. Under a spatially correlated perturbation of the same amplitude the separation survives on four of the five cases, the wave case degrading alike under both formats. The larger degradation of M2 is consistent with its fitted object, more than an order of magnitude larger than that of M1.

\label{app:robustness_box}
The last check pushes Case~V toward the boundary of the reference neighborhood. The reported transverse design box has width $0.20$. Table~\ref{tab:app_box} widens it to $0.60$ with everything else fixed: the seed, the interior split, the ranks $(5,2,2,5)$, $k=5$, $r_p=4$ and $r_s=6$.
\begin{table}[!htb]
\centering
\caption{Case~V as the transverse design box widens, at fixed ranks, split and neighborhood size, with angles in units of $\pi/2$ and errors in percent.}
\label{tab:app_box}
\footnotesize
\setlength{\tabcolsep}{5pt}
\begin{tabular}{lccrrr}
\toprule
& \multicolumn{2}{c}{Max training angle} & \multicolumn{2}{c}{Test error} & \\
\cmidrule(lr){2-3}\cmidrule(lr){4-5}
Box width & M1 ($y$ / $z$) & M2 & M1 & M2 & Gap (pp) \\
\midrule
0.20 & 0.615 / 0.589 & 0.644 & 6.90 & 6.69 & $+0.20$ \\
0.30 & 0.998 / 0.903 & 0.783 & 8.42 & 7.80 & $+0.62$ \\
0.40 & 0.979 / 0.919 & 0.828 & 15.63 & 10.32 & $+5.30$ \\
0.50 & 0.951 / 0.938 & 0.848 & 15.69 & 12.59 & $+3.11$ \\
0.60 & 0.995 / 0.932 & 0.819 & 16.41 & 14.13 & $+2.28$ \\
\bottomrule
\end{tabular}
\end{table}
From a box of $0.30$ on, training runs at opposite ends have transverse rank-two subspaces nearly orthogonal to the reference (largest M1 training angle $0.95$ to $0.998\,(\pi/2)$ against $0.78$ to $0.85$ for the rank-six subspace of M2). The safeguard never rescales a prediction, since the largest predicted angle stays at $0.81\,(\pi/2)$ for M1 and $0.68$ for M2. Both methods degrade as the box widens, M1 from $6.90$ to $16.41$ percent and M2 from $6.69$ to $14.13$. The fixed-core ablation doubles, from $19.12$ to $38.98$ percent, and the gap opens from $0.20$ points to $0.62$, $5.30$, $3.11$ and $2.28$ points at widths $0.30$ to $0.60$, with the training truncation error fixed at $1.9$ to $2.0$ percent. The mode-wise construction loses faster than the single-subspace one once its low-rank transverse subspaces turn through most of a right angle between neighboring runs, the regime that the extension with several local references is meant to address. A census of the safeguard over the five original and fifty resampled splits finds it active on three of the $616$ predictions, all on two resampled draws of the external case whose training angles pass $0.97\,(\pi/2)$, with a largest pre-clip predicted angle of $0.92\,(\pi/2)$.

{\fontsize{9}{10}\selectfont\setlength{\bibsep}{0.02ex plus 0.02ex}\setlength{\emergencystretch}{1.5em}
}


\begin{thebibliography}{105}
\expandafter\ifx\csname natexlab\endcsname\relax\def\natexlab#1{#1}\fi
\providecommand{\url}[1]{\texttt{#1}}
\providecommand{\href}[2]{#2}
\providecommand{\path}[1]{#1}
\providecommand{\DOIprefix}{doi:}
\providecommand{\ArXivprefix}{arXiv:}
\providecommand{\URLprefix}{URL: }
\providecommand{\Pubmedprefix}{pmid:}
\providecommand{\doi}[1]{\href{http://dx.doi.org/#1}{\path{#1}}}
\providecommand{\Pubmed}[1]{\href{pmid:#1}{\path{#1}}}
\providecommand{\bibinfo}[2]{#2}
\ifx\xfnm\relax \def\xfnm[#1]{\unskip,\space#1}\fi
\bibitem[{Deo et~al.(2026)Deo, Venkateshwaran, and
  Jaiman}]{deo2025physicsguided}
\bibinfo{author}{I.~K. Deo}, \bibinfo{author}{A.~Venkateshwaran},
  \bibinfo{author}{R.~K. Jaiman},
\newblock \bibinfo{title}{A physics-guided probabilistic surrogate modeling
  framework for digital twins of underwater radiated noise},
\newblock \bibinfo{journal}{Ocean Engineering} \bibinfo{volume}{358}
  (\bibinfo{year}{2026}) \bibinfo{pages}{125461}.
  \DOIprefix\doi{10.1016/j.oceaneng.2026.125461}.
\bibitem[{Musekamp et~al.(2025)Musekamp, Kalimuthu, Holzm{\"u}ller, Takamoto,
  and Niepert}]{musekamp2025al4pde}
\bibinfo{author}{D.~Musekamp}, \bibinfo{author}{M.~Kalimuthu},
  \bibinfo{author}{D.~Holzm{\"u}ller}, \bibinfo{author}{M.~Takamoto},
  \bibinfo{author}{M.~Niepert},
\newblock \bibinfo{title}{Active learning for neural {PDE} solvers},
\newblock in: \bibinfo{booktitle}{International Conference on Learning
  Representations (ICLR)}, \bibinfo{year}{2025}.
  \bibinfo{note}{{arXiv}:2408.01536}.
\bibitem[{Tian et~al.(2026)Tian, Mu, Ding, Li, Ding, and
  Zhao}]{tian2025distortion}
\bibinfo{author}{M.~Tian}, \bibinfo{author}{H.~Mu}, \bibinfo{author}{D.~Ding},
  \bibinfo{author}{M.~Li}, \bibinfo{author}{Y.~Ding},
  \bibinfo{author}{J.~Zhao},
\newblock \bibinfo{title}{Real-time distortion prediction in metallic additive
  manufacturing via a physics-informed deep operator network},
\newblock \bibinfo{journal}{Journal of Industrial Information Integration}
  \bibinfo{volume}{52} (\bibinfo{year}{2026}) \bibinfo{pages}{101126}.
  \DOIprefix\doi{10.1016/j.jii.2026.101126}.
\bibitem[{Tian et~al.(2025)Tian, Mu, Liu, Li, Ding, and
  Zhao}]{tian2025temperature}
\bibinfo{author}{M.~Tian}, \bibinfo{author}{H.~Mu}, \bibinfo{author}{T.~Liu},
  \bibinfo{author}{M.~Li}, \bibinfo{author}{D.~Ding},
  \bibinfo{author}{J.~Zhao},
\newblock \bibinfo{title}{Physics-informed machine learning-based real-time
  long-horizon temperature fields prediction in metallic additive
  manufacturing},
\newblock \bibinfo{journal}{Communications Engineering} \bibinfo{volume}{4}
  (\bibinfo{year}{2025}) \bibinfo{pages}{168}.
  \DOIprefix\doi{10.1038/s44172-025-00501-7}.
\bibitem[{Thiamchaiboonthawee et~al.(2026)Thiamchaiboonthawee, Nehme,
  Telikicherla, Tian, Jayaraman, Chandan, Mariappan, and
  Ahmed}]{thiamchaiboonthawee2026flare}
\bibinfo{author}{K.~Thiamchaiboonthawee}, \bibinfo{author}{G.~Nehme},
  \bibinfo{author}{R.~M. Telikicherla}, \bibinfo{author}{J.~Tian},
  \bibinfo{author}{B.~Jayaraman}, \bibinfo{author}{V.~Chandan},
  \bibinfo{author}{D.~Mariappan}, \bibinfo{author}{F.~Ahmed},
  \bibinfo{title}{{FLARE}: A data-efficient surrogate for predicting
  displacement fields in directed energy deposition}, \bibinfo{year}{2026}.
  \href{http://arxiv.org/abs/2604.16649}{{\tt arXiv:2604.16649}}.
\bibitem[{Giacomini and D{\'i}ez(2026)}]{giacomini2026surrogates}
\bibinfo{author}{M.~Giacomini}, \bibinfo{author}{P.~D{\'i}ez},
\newblock \bibinfo{title}{Surrogates for physics-based and data-driven
  modelling of parametric systems: Review and new perspectives},
\newblock \bibinfo{journal}{Archives of Computational Methods in Engineering}
  \bibinfo{volume}{33} (\bibinfo{year}{2026}) \bibinfo{pages}{8519--8551}.
  \DOIprefix\doi{10.1007/s11831-026-10552-4}.
\bibitem[{Mosquera et~al.(2019)Mosquera, Hamdouni, El~Hamidi, and
  Allery}]{mosquera2019idw}
\bibinfo{author}{R.~Mosquera}, \bibinfo{author}{A.~Hamdouni},
  \bibinfo{author}{A.~El~Hamidi}, \bibinfo{author}{C.~Allery},
\newblock \bibinfo{title}{{POD} basis interpolation via inverse distance
  weighting on {G}rassmann manifolds},
\newblock \bibinfo{journal}{Discrete and Continuous Dynamical Systems - S}
  \bibinfo{volume}{12} (\bibinfo{year}{2019}) \bibinfo{pages}{1743--1759}.
  \DOIprefix\doi{10.3934/dcdss.2019115}.
\bibitem[{Oulghelou and Allery(2021)}]{oulghelou2021bicitsgm}
\bibinfo{author}{M.~Oulghelou}, \bibinfo{author}{C.~Allery},
\newblock \bibinfo{title}{Non-intrusive method for parametric model order
  reduction using a bi-calibrated interpolation on the {G}rassmann manifold},
\newblock \bibinfo{journal}{Journal of Computational Physics}
  \bibinfo{volume}{426} (\bibinfo{year}{2021}) \bibinfo{pages}{109924}.
  \DOIprefix\doi{10.1016/j.jcp.2020.109924}.
\bibitem[{Friderikos et~al.(2022)Friderikos, Baranger, Olive, and
  N{\'e}ron}]{friderikos2022stability}
\bibinfo{author}{O.~Friderikos}, \bibinfo{author}{E.~Baranger},
  \bibinfo{author}{M.~Olive}, \bibinfo{author}{D.~N{\'e}ron},
\newblock \bibinfo{title}{On the stability of {POD} basis interpolation on
  {G}rassmann manifolds for parametric model order reduction},
\newblock \bibinfo{journal}{Computational Mechanics} \bibinfo{volume}{70}
  (\bibinfo{year}{2022}) \bibinfo{pages}{181--204}.
  \DOIprefix\doi{10.1007/s00466-022-02163-0}.
\bibitem[{De~Lathauwer et~al.(2000)De~Lathauwer, De~Moor, and
  Vandewalle}]{DeLathauwer2000a}
\bibinfo{author}{L.~De~Lathauwer}, \bibinfo{author}{B.~De~Moor},
  \bibinfo{author}{J.~Vandewalle},
\newblock \bibinfo{title}{A multilinear singular value decomposition},
\newblock \bibinfo{journal}{SIAM Journal on Matrix Analysis and Applications}
  \bibinfo{volume}{21} (\bibinfo{year}{2000}) \bibinfo{pages}{1253--1278}.
  \DOIprefix\doi{10.1137/S0895479896305696}.
\bibitem[{Kolda and Bader(2009)}]{Kolda2009}
\bibinfo{author}{T.~G. Kolda}, \bibinfo{author}{B.~W. Bader},
\newblock \bibinfo{title}{Tensor decompositions and applications},
\newblock \bibinfo{journal}{SIAM Review} \bibinfo{volume}{51}
  (\bibinfo{year}{2009}) \bibinfo{pages}{455--500}.
  \DOIprefix\doi{10.1137/07070111X}.
\bibitem[{Sch{\"o}nemann(1966)}]{Schoenemann1966}
\bibinfo{author}{P.~H. Sch{\"o}nemann},
\newblock \bibinfo{title}{A generalized solution of the orthogonal {Procrustes}
  problem},
\newblock \bibinfo{journal}{Psychometrika} \bibinfo{volume}{31}
  (\bibinfo{year}{1966}) \bibinfo{pages}{1--10}.
  \DOIprefix\doi{10.1007/BF02289451}.
\bibitem[{Higham(1986)}]{Higham1986}
\bibinfo{author}{N.~J. Higham},
\newblock \bibinfo{title}{Computing the polar decomposition---with
  applications},
\newblock \bibinfo{journal}{SIAM Journal on Scientific and Statistical
  Computing} \bibinfo{volume}{7} (\bibinfo{year}{1986})
  \bibinfo{pages}{1160--1174}. \DOIprefix\doi{10.1137/0907079}.
\bibitem[{Buhmann(2003)}]{Buhmann2003}
\bibinfo{author}{M.~D. Buhmann}, \bibinfo{title}{Radial Basis Functions: Theory
  and Implementations}, \bibinfo{publisher}{Cambridge University Press},
  \bibinfo{year}{2003}. \DOIprefix\doi{10.1017/CBO9780511543241}.
\bibitem[{Rasmussen and Williams(2006)}]{Rasmussen2006}
\bibinfo{author}{C.~E. Rasmussen}, \bibinfo{author}{C.~K.~I. Williams},
  \bibinfo{title}{Gaussian Processes for Machine Learning},
  \bibinfo{publisher}{MIT Press}, \bibinfo{year}{2006}.
  \DOIprefix\doi{10.7551/mitpress/3206.001.0001}.
\bibitem[{Benner et~al.(2015)Benner, Gugercin, and Willcox}]{Benner2015}
\bibinfo{author}{P.~Benner}, \bibinfo{author}{S.~Gugercin},
  \bibinfo{author}{K.~Willcox},
\newblock \bibinfo{title}{A survey of projection-based model reduction methods
  for parametric dynamical systems},
\newblock \bibinfo{journal}{SIAM Review} \bibinfo{volume}{57}
  (\bibinfo{year}{2015}) \bibinfo{pages}{483--531}.
  \DOIprefix\doi{10.1137/130932715}.
\bibitem[{Quarteroni et~al.(2016)Quarteroni, Manzoni, and
  Negri}]{Quarteroni2016}
\bibinfo{author}{A.~Quarteroni}, \bibinfo{author}{A.~Manzoni},
  \bibinfo{author}{F.~Negri}, \bibinfo{title}{Reduced Basis Methods for Partial
  Differential Equations: An Introduction}, volume~\bibinfo{volume}{92} of
  \textit{\bibinfo{series}{UNITEXT}}, \bibinfo{publisher}{Springer},
  \bibinfo{address}{Cham}, \bibinfo{year}{2016}.
  \DOIprefix\doi{10.1007/978-3-319-15431-2}.
\bibitem[{Hesthaven et~al.(2016)Hesthaven, Rozza, and Stamm}]{Hesthaven2016}
\bibinfo{author}{J.~S. Hesthaven}, \bibinfo{author}{G.~Rozza},
  \bibinfo{author}{B.~Stamm}, \bibinfo{title}{Certified Reduced Basis Methods
  for Parametrized Partial Differential Equations}, SpringerBriefs in
  Mathematics, \bibinfo{publisher}{Springer}, \bibinfo{year}{2016}.
  \DOIprefix\doi{10.1007/978-3-319-22470-1}.
\bibitem[{Benner et~al.(2017)Benner, Ohlberger, Cohen, and
  Willcox}]{Benner2017Book}
\bibinfo{editor}{P.~Benner}, \bibinfo{editor}{M.~Ohlberger},
  \bibinfo{editor}{A.~Cohen}, \bibinfo{editor}{K.~Willcox} (Eds.),
  \bibinfo{title}{Model Reduction and Approximation: Theory and Algorithms},
  \bibinfo{publisher}{SIAM}, \bibinfo{address}{Philadelphia},
  \bibinfo{year}{2017}. \DOIprefix\doi{10.1137/1.9781611974829}.
\bibitem[{Sirovich(1987)}]{Sirovich1987}
\bibinfo{author}{L.~Sirovich},
\newblock \bibinfo{title}{Turbulence and the dynamics of coherent structures.
  part {I}: Coherent structures},
\newblock \bibinfo{journal}{Quarterly of Applied Mathematics}
  \bibinfo{volume}{45} (\bibinfo{year}{1987}) \bibinfo{pages}{561--571}.
  \DOIprefix\doi{10.1090/qam/910462}.
\bibitem[{Berkooz et~al.(1993)Berkooz, Holmes, and Lumley}]{Berkooz1993}
\bibinfo{author}{G.~Berkooz}, \bibinfo{author}{P.~Holmes},
  \bibinfo{author}{J.~L. Lumley},
\newblock \bibinfo{title}{The proper orthogonal decomposition in the analysis
  of turbulent flows},
\newblock \bibinfo{journal}{Annual Review of Fluid Mechanics}
  \bibinfo{volume}{25} (\bibinfo{year}{1993}) \bibinfo{pages}{539--575}.
  \DOIprefix\doi{10.1146/annurev.fl.25.010193.002543}.
\bibitem[{Willcox and Peraire(2002)}]{Willcox2002}
\bibinfo{author}{K.~Willcox}, \bibinfo{author}{J.~Peraire},
\newblock \bibinfo{title}{Balanced model reduction via the proper orthogonal
  decomposition},
\newblock \bibinfo{journal}{AIAA Journal} \bibinfo{volume}{40}
  (\bibinfo{year}{2002}) \bibinfo{pages}{2323--2330}.
  \DOIprefix\doi{10.2514/2.1570}.
\bibitem[{Rowley(2005)}]{Rowley2004}
\bibinfo{author}{C.~W. Rowley},
\newblock \bibinfo{title}{Model reduction for fluids using balanced proper
  orthogonal decomposition},
\newblock \bibinfo{journal}{International Journal of Bifurcation and Chaos in
  Applied Sciences and Engineering} \bibinfo{volume}{15} (\bibinfo{year}{2005})
  \bibinfo{pages}{997--1013}. \DOIprefix\doi{10.1142/S0218127405012429}.
\bibitem[{Amsallem et~al.(2012)Amsallem, Zahr, and
  Farhat}]{Amsallem2012LocalROM}
\bibinfo{author}{D.~Amsallem}, \bibinfo{author}{M.~J. Zahr},
  \bibinfo{author}{C.~Farhat},
\newblock \bibinfo{title}{Nonlinear model order reduction based on local
  reduced-order bases},
\newblock \bibinfo{journal}{International Journal for Numerical Methods in
  Engineering} \bibinfo{volume}{92} (\bibinfo{year}{2012})
  \bibinfo{pages}{891--916}. \DOIprefix\doi{10.1002/nme.4371}.
\bibitem[{Edelman et~al.(1998)Edelman, Arias, and Smith}]{Edelman1998}
\bibinfo{author}{A.~Edelman}, \bibinfo{author}{T.~A. Arias},
  \bibinfo{author}{S.~T. Smith},
\newblock \bibinfo{title}{The geometry of algorithms with orthogonality
  constraints},
\newblock \bibinfo{journal}{SIAM Journal on Matrix Analysis and Applications}
  \bibinfo{volume}{20} (\bibinfo{year}{1998}) \bibinfo{pages}{303--353}.
  \DOIprefix\doi{10.1137/S0895479895290954}.
\bibitem[{Absil et~al.(2008)Absil, Mahony, and Sepulchre}]{Absil2008}
\bibinfo{author}{P.-A. Absil}, \bibinfo{author}{R.~Mahony},
  \bibinfo{author}{R.~Sepulchre}, \bibinfo{title}{Optimization Algorithms on
  Matrix Manifolds}, \bibinfo{publisher}{Princeton University Press},
  \bibinfo{address}{Princeton, NJ}, \bibinfo{year}{2008}.
  \DOIprefix\doi{10.1515/9781400830244}.
\bibitem[{Kobayashi and Nomizu(1963)}]{Kobayashi1963}
\bibinfo{author}{S.~Kobayashi}, \bibinfo{author}{K.~Nomizu},
  \bibinfo{title}{Foundations of Differential Geometry},
  volume~\bibinfo{volume}{1}, \bibinfo{publisher}{Interscience},
  \bibinfo{year}{1963}.
\bibitem[{Amsallem and Farhat(2008)}]{Amsallem2008}
\bibinfo{author}{D.~Amsallem}, \bibinfo{author}{C.~Farhat},
\newblock \bibinfo{title}{Interpolation method for adapting reduced-order
  models and application to aeroelasticity},
\newblock \bibinfo{journal}{AIAA Journal} \bibinfo{volume}{46}
  (\bibinfo{year}{2008}) \bibinfo{pages}{1803--1813}.
  \DOIprefix\doi{10.2514/1.35374}.
\bibitem[{Son(2013)}]{Son2013}
\bibinfo{author}{N.~T. Son},
\newblock \bibinfo{title}{A real time procedure for affinely dependent
  parametric model order reduction using interpolation on {G}rassmann
  manifolds},
\newblock \bibinfo{journal}{International Journal for Numerical Methods in
  Engineering} \bibinfo{volume}{93} (\bibinfo{year}{2013})
  \bibinfo{pages}{818--833}. \DOIprefix\doi{10.1002/nme.4408}.
\bibitem[{Amsallem et~al.(2009)Amsallem, Cortial, Carlberg, and
  Farhat}]{amsallem2009interpolating}
\bibinfo{author}{D.~Amsallem}, \bibinfo{author}{J.~Cortial},
  \bibinfo{author}{K.~Carlberg}, \bibinfo{author}{C.~Farhat},
\newblock \bibinfo{title}{A method for interpolating on manifolds structural
  dynamics reduced-order models},
\newblock \bibinfo{journal}{International Journal for Numerical Methods in
  Engineering} \bibinfo{volume}{80} (\bibinfo{year}{2009})
  \bibinfo{pages}{1241--1258}. \DOIprefix\doi{10.1002/nme.2681}.
\bibitem[{Amsallem and Farhat(2011)}]{Amsallem2011}
\bibinfo{author}{D.~Amsallem}, \bibinfo{author}{C.~Farhat},
\newblock \bibinfo{title}{An online method for interpolating linear parametric
  reduced-order models},
\newblock \bibinfo{journal}{SIAM Journal on Scientific Computing}
  \bibinfo{volume}{33} (\bibinfo{year}{2011}) \bibinfo{pages}{2169--2198}.
  \DOIprefix\doi{10.1137/100813051}.
\bibitem[{Bendokat et~al.(2024)Bendokat, Zimmermann, and Absil}]{Bendokat2020}
\bibinfo{author}{T.~Bendokat}, \bibinfo{author}{R.~Zimmermann},
  \bibinfo{author}{P.-A. Absil},
\newblock \bibinfo{title}{A {G}rassmann manifold handbook: Basic geometry and
  computational aspects},
\newblock \bibinfo{journal}{Advances in Computational Mathematics}
  \bibinfo{volume}{50} (\bibinfo{year}{2024}) \bibinfo{pages}{6}.
  \DOIprefix\doi{10.1007/s10444-023-10090-8}.
\bibitem[{Zimmermann(2021)}]{zimmermann2021manifold}
\bibinfo{author}{R.~Zimmermann},
\newblock \bibinfo{title}{Manifold interpolation},
\newblock in: \bibinfo{editor}{P.~Benner}, \bibinfo{editor}{S.~Grivet-Talocia},
  \bibinfo{editor}{A.~Quarteroni}, \bibinfo{editor}{G.~Rozza},
  \bibinfo{editor}{W.~Schilders}, \bibinfo{editor}{L.~M. Silveira} (Eds.),
  \bibinfo{booktitle}{Model Order Reduction, Volume 1: System- and Data-Driven
  Methods and Algorithms}, \bibinfo{publisher}{De Gruyter},
  \bibinfo{address}{Berlin, Boston}, \bibinfo{year}{2021}, pp.
  \bibinfo{pages}{229--274}. \DOIprefix\doi{10.1515/9783110498967-007}.
\bibitem[{Zhang et~al.(2022)Zhang, Mak, and Dunson}]{Zhang2022GPS}
\bibinfo{author}{R.~Zhang}, \bibinfo{author}{S.~Mak},
  \bibinfo{author}{D.~Dunson},
\newblock \bibinfo{title}{Gaussian process subspace prediction for model
  reduction},
\newblock \bibinfo{journal}{SIAM Journal on Scientific Computing}
  \bibinfo{volume}{44} (\bibinfo{year}{2022}) \bibinfo{pages}{A1428--A1449}.
  \DOIprefix\doi{10.1137/21M1432739}.
\bibitem[{Ciaramella et~al.(2026)Ciaramella, Gander, and
  Vanzan}]{ciaramella2026grassmann}
\bibinfo{author}{G.~Ciaramella}, \bibinfo{author}{M.~J. Gander},
  \bibinfo{author}{T.~Vanzan},
\newblock \bibinfo{title}{A gentle introduction to interpolation on the
  {G}rassmann manifold},
\newblock \bibinfo{journal}{SIAM Review} \bibinfo{volume}{68}
  (\bibinfo{year}{2026}) \bibinfo{pages}{172--203}.
  \DOIprefix\doi{10.1137/24M1628591}.
\bibitem[{Sreekumar et~al.(2026)Sreekumar, Komis, Barman, Astolfi, and
  Chronopoulos}]{sreekumar2026roadmap}
\bibinfo{author}{A.~Sreekumar}, \bibinfo{author}{G.~Komis},
  \bibinfo{author}{S.~K. Barman}, \bibinfo{author}{A.~Astolfi},
  \bibinfo{author}{D.~Chronopoulos},
\newblock \bibinfo{title}{From snapshots to manifolds: a practical roadmap for
  interpolatory reduced-order modelling},
\newblock \bibinfo{journal}{Proceedings of the Royal Society A}
  \bibinfo{volume}{482} (\bibinfo{year}{2026}) \bibinfo{pages}{20251004}.
  \DOIprefix\doi{10.1098/rspa.2025.1004}.
\bibitem[{Sato and Schmidt(2025)}]{sato2025parametric}
\bibinfo{author}{S.~Sato}, \bibinfo{author}{O.~T. Schmidt},
\newblock \bibinfo{title}{Parametric reduced-order modelling and mode
  sensitivity of actuated cylinder flow from a matrix manifold perspective},
\newblock \bibinfo{journal}{Journal of Fluid Mechanics} \bibinfo{volume}{1021}
  (\bibinfo{year}{2025}) \bibinfo{pages}{A44}.
  \DOIprefix\doi{10.1017/jfm.2025.10733}.
\bibitem[{Mosquera et~al.(2021)Mosquera, El~Hamidi, Hamdouni, and
  Falaize}]{mosquera2021neville}
\bibinfo{author}{R.~Mosquera}, \bibinfo{author}{A.~El~Hamidi},
  \bibinfo{author}{A.~Hamdouni}, \bibinfo{author}{A.~Falaize},
\newblock \bibinfo{title}{Generalization of the {N}eville--{A}itken
  interpolation algorithm on {G}rassmann manifolds: Applications to reduced
  order model},
\newblock \bibinfo{journal}{International Journal for Numerical Methods in
  Fluids} \bibinfo{volume}{93} (\bibinfo{year}{2021})
  \bibinfo{pages}{2421--2442}. \DOIprefix\doi{10.1002/fld.4981}.
\bibitem[{Niu et~al.(2026)Niu, Jiang, Fei, Xiong, and Li}]{niu2026stable}
\bibinfo{author}{Q.~Niu}, \bibinfo{author}{W.~Jiang}, \bibinfo{author}{J.~Fei},
  \bibinfo{author}{R.~Xiong}, \bibinfo{author}{Y.~Li}, \bibinfo{title}{Stable
  high-order interpolation on the {G}rassmann manifold by maximum-volume
  coordinates and {A}rnoldi orthogonalization}, \bibinfo{year}{2026}.
  \href{http://arxiv.org/abs/2603.23858}{{\tt arXiv:2603.23858}}.
\bibitem[{Jensen and Zimmermann(2025)}]{jensen2025maxvol}
\bibinfo{author}{R.~Jensen}, \bibinfo{author}{R.~Zimmermann},
  \bibinfo{title}{Maximum volume coordinates for {G}rassmann interpolation:
  {L}agrange, {H}ermite, and errors}, \bibinfo{year}{2025}.
  \href{http://arxiv.org/abs/2506.01574}{{\tt arXiv:2506.01574}}.
\bibitem[{Oulghelou et~al.(2021)Oulghelou, Allery, and
  Mosquera}]{oulghelou2021barycentric}
\bibinfo{author}{M.~Oulghelou}, \bibinfo{author}{C.~Allery},
  \bibinfo{author}{R.~Mosquera},
\newblock \bibinfo{title}{Parametric reduced order models based on a
  {R}iemannian barycentric interpolation},
\newblock \bibinfo{journal}{International Journal for Numerical Methods in
  Engineering} \bibinfo{volume}{122} (\bibinfo{year}{2021})
  \bibinfo{pages}{6623--6640}. \DOIprefix\doi{10.1002/nme.6805}.
\bibitem[{El~Omari et~al.(2025)El~Omari, El~Khlifi, and
  Cordier}]{elomari2025stiefel}
\bibinfo{author}{A.~El~Omari}, \bibinfo{author}{M.~El~Khlifi},
  \bibinfo{author}{L.~Cordier},
\newblock \bibinfo{title}{Stiefel manifold interpolation for non-intrusive
  model reduction of parameterized fluid flow problems},
\newblock \bibinfo{journal}{Journal of Computational Physics}
  \bibinfo{volume}{521} (\bibinfo{year}{2025}) \bibinfo{pages}{113564}.
  \DOIprefix\doi{10.1016/j.jcp.2024.113564}.
\bibitem[{Mosquera et~al.(2024)Mosquera, Falaize, and
  Hamdouni}]{mosquera2024kriging}
\bibinfo{author}{R.~Mosquera}, \bibinfo{author}{A.~Falaize},
  \bibinfo{author}{A.~Hamdouni},
\newblock \bibinfo{title}{Grassmannian kriging with applications in {POD}-based
  model order reduction},
\newblock \bibinfo{journal}{Discrete and Continuous Dynamical Systems - S}
  \bibinfo{volume}{17} (\bibinfo{year}{2024}) \bibinfo{pages}{2400--2419}.
  \DOIprefix\doi{10.3934/dcdss.2024017}.
\bibitem[{Giovanis and Shields(2020)}]{Giovanis2020GrassmannGP}
\bibinfo{author}{D.~G. Giovanis}, \bibinfo{author}{M.~D. Shields},
\newblock \bibinfo{title}{Data-driven surrogates for high dimensional models
  using {G}aussian process regression on the {G}rassmann manifold},
\newblock \bibinfo{journal}{Computer Methods in Applied Mechanics and
  Engineering} \bibinfo{volume}{370} (\bibinfo{year}{2020})
  \bibinfo{pages}{113269}. \DOIprefix\doi{10.1016/j.cma.2020.113269}.
\bibitem[{Liu and Liu(2023)}]{liu2023regression}
\bibinfo{author}{X.~Liu}, \bibinfo{author}{X.~Liu},
\newblock \bibinfo{title}{Regression trees on {G}rassmann manifold for adapting
  reduced-order models},
\newblock \bibinfo{journal}{AIAA Journal} \bibinfo{volume}{61}
  (\bibinfo{year}{2023}) \bibinfo{pages}{1318--1333}.
  \DOIprefix\doi{10.2514/1.J062180}.
\bibitem[{Liu et~al.(2026)Liu, Feng, and Liu}]{liu2026projected}
\bibinfo{author}{X.~Liu}, \bibinfo{author}{J.~Feng}, \bibinfo{author}{X.~Liu},
\newblock \bibinfo{title}{Adapting projection-based reduced-order models using
  projected {G}aussian process},
\newblock \bibinfo{journal}{SIAM Journal on Scientific Computing}
  \bibinfo{volume}{48} (\bibinfo{year}{2026}) \bibinfo{pages}{C941--C966}.
  \DOIprefix\doi{10.1137/24M1697827}.
\bibitem[{Giovanis and Shields(2018)}]{giovanis2018uncertainty}
\bibinfo{author}{D.~G. Giovanis}, \bibinfo{author}{M.~D. Shields},
\newblock \bibinfo{title}{Uncertainty quantification for complex systems with
  very high dimensional response using {G}rassmann manifold variations},
\newblock \bibinfo{journal}{Journal of Computational Physics}
  \bibinfo{volume}{364} (\bibinfo{year}{2018}) \bibinfo{pages}{393--415}.
  \DOIprefix\doi{10.1016/j.jcp.2018.03.009}.
\bibitem[{Du et~al.(2026)Du, Zhang, Zhang, and Zhang}]{du2026dmd}
\bibinfo{author}{L.~Du}, \bibinfo{author}{S.~Zhang},
  \bibinfo{author}{R.~Zhang}, \bibinfo{author}{S.~Zhang},
\newblock \bibinfo{title}{Interpolation-based parametric reduced-order models
  with dynamic mode decomposition},
\newblock \bibinfo{journal}{Journal of Computational Physics}
  \bibinfo{volume}{544} (\bibinfo{year}{2026}) \bibinfo{pages}{114436}.
  \DOIprefix\doi{10.1016/j.jcp.2025.114436}.
\bibitem[{Peherstorfer and Willcox(2016)}]{peherstorfer2016opinf}
\bibinfo{author}{B.~Peherstorfer}, \bibinfo{author}{K.~Willcox},
\newblock \bibinfo{title}{Data-driven operator inference for nonintrusive
  projection-based model reduction},
\newblock \bibinfo{journal}{Computer Methods in Applied Mechanics and
  Engineering} \bibinfo{volume}{306} (\bibinfo{year}{2016})
  \bibinfo{pages}{196--215}. \DOIprefix\doi{10.1016/j.cma.2016.03.025}.
\bibitem[{Li et~al.(2021)Li, Kovachki, Azizzadenesheli, Liu, Bhattacharya,
  Stuart, and Anandkumar}]{li2021fourier}
\bibinfo{author}{Z.~Li}, \bibinfo{author}{N.~Kovachki},
  \bibinfo{author}{K.~Azizzadenesheli}, \bibinfo{author}{B.~Liu},
  \bibinfo{author}{K.~Bhattacharya}, \bibinfo{author}{A.~Stuart},
  \bibinfo{author}{A.~Anandkumar},
\newblock \bibinfo{title}{Fourier neural operator for parametric partial
  differential equations},
\newblock in: \bibinfo{booktitle}{International Conference on Learning
  Representations (ICLR)}, \bibinfo{year}{2021}.
  \bibinfo{note}{{arXiv}:2010.08895}.
\bibitem[{Lu et~al.(2021)Lu, Jin, Pang, Zhang, and
  Karniadakis}]{lu2021deeponet}
\bibinfo{author}{L.~Lu}, \bibinfo{author}{P.~Jin}, \bibinfo{author}{G.~Pang},
  \bibinfo{author}{Z.~Zhang}, \bibinfo{author}{G.~E. Karniadakis},
\newblock \bibinfo{title}{Learning nonlinear operators via {DeepONet} based on
  the universal approximation theorem of operators},
\newblock \bibinfo{journal}{Nature Machine Intelligence} \bibinfo{volume}{3}
  (\bibinfo{year}{2021}) \bibinfo{pages}{218--229}.
  \DOIprefix\doi{10.1038/s42256-021-00302-5}.
\bibitem[{Li et~al.(2023)Li, Shu, and Barati~Farimani}]{li2023scalable}
\bibinfo{author}{Z.~Li}, \bibinfo{author}{D.~Shu},
  \bibinfo{author}{A.~Barati~Farimani},
\newblock \bibinfo{title}{Scalable transformer for {PDE} surrogate modeling},
\newblock in: \bibinfo{booktitle}{Advances in Neural Information Processing
  Systems (NeurIPS)}, volume~\bibinfo{volume}{36}, \bibinfo{year}{2023}, pp.
  \bibinfo{pages}{28010--28039}. \DOIprefix\doi{10.52202/075280-1216}.
\bibitem[{Wu et~al.(2024)Wu, Luo, Wang, Wang, and Long}]{wu2024transolver}
\bibinfo{author}{H.~Wu}, \bibinfo{author}{H.~Luo}, \bibinfo{author}{H.~Wang},
  \bibinfo{author}{J.~Wang}, \bibinfo{author}{M.~Long},
\newblock \bibinfo{title}{Transolver: A fast transformer solver for {PDE}s on
  general geometries},
\newblock in: \bibinfo{booktitle}{Proceedings of the 41st International
  Conference on Machine Learning}, volume \bibinfo{volume}{235} of
  \textit{\bibinfo{series}{Proceedings of Machine Learning Research}},
  \bibinfo{year}{2024}, pp. \bibinfo{pages}{53681--53705}.
  \bibinfo{note}{{arXiv}:2402.02366}.
\bibitem[{Luo et~al.(2025)Luo, Wu, Zhou, Xing, Di, Wang, and
  Long}]{luo2025transolverpp}
\bibinfo{author}{H.~Luo}, \bibinfo{author}{H.~Wu}, \bibinfo{author}{H.~Zhou},
  \bibinfo{author}{L.~Xing}, \bibinfo{author}{Y.~Di},
  \bibinfo{author}{J.~Wang}, \bibinfo{author}{M.~Long},
\newblock \bibinfo{title}{Transolver++: An accurate neural solver for {PDE}s on
  million-scale geometries},
\newblock in: \bibinfo{booktitle}{Proceedings of the 42nd International
  Conference on Machine Learning}, volume \bibinfo{volume}{267} of
  \textit{\bibinfo{series}{Proceedings of Machine Learning Research}},
  \bibinfo{year}{2025}, pp. \bibinfo{pages}{41432--41449}.
\bibitem[{Wen et~al.(2025)Wen, Kumbhat, Lingsch, Mousavi, Zhao, Chandrashekar,
  and Mishra}]{wen2025gaot}
\bibinfo{author}{S.~Wen}, \bibinfo{author}{A.~Kumbhat},
  \bibinfo{author}{L.~Lingsch}, \bibinfo{author}{S.~Mousavi},
  \bibinfo{author}{Y.~Zhao}, \bibinfo{author}{P.~Chandrashekar},
  \bibinfo{author}{S.~Mishra},
\newblock \bibinfo{title}{Geometry aware operator transformer as an efficient
  and accurate neural surrogate for {PDE}s on arbitrary domains},
\newblock in: \bibinfo{booktitle}{Advances in Neural Information Processing
  Systems (NeurIPS)}, volume~\bibinfo{volume}{38}, \bibinfo{year}{2025}, pp.
  \bibinfo{pages}{155423--155501}. \bibinfo{note}{{arXiv}:2505.18781}.
\bibitem[{Fresca and Manzoni(2022)}]{fresca2022poddlrom}
\bibinfo{author}{S.~Fresca}, \bibinfo{author}{A.~Manzoni},
\newblock \bibinfo{title}{{POD-DL-ROM}: Enhancing deep learning-based reduced
  order models for nonlinear parametrized {PDE}s by proper orthogonal
  decomposition},
\newblock \bibinfo{journal}{Computer Methods in Applied Mechanics and
  Engineering} \bibinfo{volume}{388} (\bibinfo{year}{2022})
  \bibinfo{pages}{114181}. \DOIprefix\doi{10.1016/j.cma.2021.114181}.
\bibitem[{De~Lathauwer et~al.(2000)De~Lathauwer, De~Moor, and
  Vandewalle}]{DeLathauwer2000b}
\bibinfo{author}{L.~De~Lathauwer}, \bibinfo{author}{B.~De~Moor},
  \bibinfo{author}{J.~Vandewalle},
\newblock \bibinfo{title}{On the best rank-1 and rank-{$(R_1, R_2, \ldots,
  R_N)$} approximation of higher-order tensors},
\newblock \bibinfo{journal}{SIAM Journal on Matrix Analysis and Applications}
  \bibinfo{volume}{21} (\bibinfo{year}{2000}) \bibinfo{pages}{1324--1342}.
  \DOIprefix\doi{10.1137/S0895479898346995}.
\bibitem[{Vannieuwenhoven et~al.(2012)Vannieuwenhoven, Vandebril, and
  Meerbergen}]{Vannieuwenhoven2012}
\bibinfo{author}{N.~Vannieuwenhoven}, \bibinfo{author}{R.~Vandebril},
  \bibinfo{author}{K.~Meerbergen},
\newblock \bibinfo{title}{A new truncation strategy for the higher-order
  singular value decomposition},
\newblock \bibinfo{journal}{SIAM Journal on Scientific Computing}
  \bibinfo{volume}{34} (\bibinfo{year}{2012}) \bibinfo{pages}{A1027--A1052}.
  \DOIprefix\doi{10.1137/110836067}.
\bibitem[{Oseledets(2011)}]{Oseledets2011}
\bibinfo{author}{I.~V. Oseledets},
\newblock \bibinfo{title}{Tensor-train decomposition},
\newblock \bibinfo{journal}{SIAM Journal on Scientific Computing}
  \bibinfo{volume}{33} (\bibinfo{year}{2011}) \bibinfo{pages}{2295--2317}.
  \DOIprefix\doi{10.1137/090752286}.
\bibitem[{Hackbusch and K{\"u}hn(2009)}]{Hackbusch2009}
\bibinfo{author}{W.~Hackbusch}, \bibinfo{author}{S.~K{\"u}hn},
\newblock \bibinfo{title}{A new scheme for the tensor representation},
\newblock \bibinfo{journal}{Journal of Fourier Analysis and Applications}
  \bibinfo{volume}{15} (\bibinfo{year}{2009}) \bibinfo{pages}{706--722}.
  \DOIprefix\doi{10.1007/s00041-009-9094-9}.
\bibitem[{Grasedyck(2010)}]{Grasedyck2010}
\bibinfo{author}{L.~Grasedyck},
\newblock \bibinfo{title}{Hierarchical singular value decomposition of
  tensors},
\newblock \bibinfo{journal}{SIAM Journal on Matrix Analysis and Applications}
  \bibinfo{volume}{31} (\bibinfo{year}{2010}) \bibinfo{pages}{2029--2054}.
  \DOIprefix\doi{10.1137/090764189}.
\bibitem[{Nouy(2017)}]{Nouy2017}
\bibinfo{author}{A.~Nouy},
\newblock \bibinfo{title}{Low-rank methods for high-dimensional approximation
  and model order reduction},
\newblock in: \bibinfo{editor}{P.~Benner}, \bibinfo{editor}{M.~Ohlberger},
  \bibinfo{editor}{A.~Cohen}, \bibinfo{editor}{K.~Willcox} (Eds.),
  \bibinfo{booktitle}{Model Reduction and Approximation: Theory and
  Algorithms}, \bibinfo{publisher}{SIAM}, \bibinfo{address}{Philadelphia},
  \bibinfo{year}{2017}, pp. \bibinfo{pages}{171--226}.
  \DOIprefix\doi{10.1137/1.9781611974829.ch4}.
\bibitem[{Kressner and Tobler(2011)}]{kressner2011lowrank}
\bibinfo{author}{D.~Kressner}, \bibinfo{author}{C.~Tobler},
\newblock \bibinfo{title}{Low-rank tensor {K}rylov subspace methods for
  parametrized linear systems},
\newblock \bibinfo{journal}{SIAM Journal on Matrix Analysis and Applications}
  \bibinfo{volume}{32} (\bibinfo{year}{2011}) \bibinfo{pages}{1288--1316}.
  \DOIprefix\doi{10.1137/100799010}.
\bibitem[{Ballani and Grasedyck(2015)}]{ballani2015hierarchical}
\bibinfo{author}{J.~Ballani}, \bibinfo{author}{L.~Grasedyck},
\newblock \bibinfo{title}{Hierarchical tensor approximation of output
  quantities of parameter-dependent {PDE}s},
\newblock \bibinfo{journal}{SIAM/ASA Journal on Uncertainty Quantification}
  \bibinfo{volume}{3} (\bibinfo{year}{2015}) \bibinfo{pages}{852--872}.
  \DOIprefix\doi{10.1137/140960980}.
\bibitem[{Mamonov and Olshanskii(2022)}]{Mamonov2022}
\bibinfo{author}{A.~V. Mamonov}, \bibinfo{author}{M.~A. Olshanskii},
\newblock \bibinfo{title}{Interpolatory tensorial reduced order models for
  parametric dynamical systems},
\newblock \bibinfo{journal}{Computer Methods in Applied Mechanics and
  Engineering} \bibinfo{volume}{397} (\bibinfo{year}{2022})
  \bibinfo{pages}{115122}. \DOIprefix\doi{10.1016/j.cma.2022.115122}.
\bibitem[{Mamonov and Olshanskii(2024)}]{Mamonov2024}
\bibinfo{author}{A.~V. Mamonov}, \bibinfo{author}{M.~A. Olshanskii},
\newblock \bibinfo{title}{Tensorial parametric model order reduction of
  nonlinear dynamical systems},
\newblock \bibinfo{journal}{SIAM Journal on Scientific Computing}
  \bibinfo{volume}{46} (\bibinfo{year}{2024}) \bibinfo{pages}{A1850--A1878}.
  \DOIprefix\doi{10.1137/23M1553789}.
\bibitem[{Olshanskii and Rebholz(2025)}]{Olshanskii2024}
\bibinfo{author}{M.~A. Olshanskii}, \bibinfo{author}{L.~G. Rebholz},
\newblock \bibinfo{title}{Approximating a branch of solutions to the
  {Navier--Stokes} equations by reduced-order modeling},
\newblock \bibinfo{journal}{Journal of Computational Physics}
  \bibinfo{volume}{524} (\bibinfo{year}{2025}) \bibinfo{pages}{113728}.
  \DOIprefix\doi{10.1016/j.jcp.2025.113728}.
\bibitem[{Islam et~al.(2026)Islam, Mizan, Olshanskii, and
  Mang}]{islam2026tensorial}
\bibinfo{author}{A.~Islam}, \bibinfo{author}{M.~R.~B. Mizan},
  \bibinfo{author}{M.~Olshanskii}, \bibinfo{author}{A.~Mang},
  \bibinfo{title}{Tensorial reduced-order models for parametric coupled
  reaction-diffusion systems: Application to brain tumor growth modeling},
  \bibinfo{year}{2026}. \href{http://arxiv.org/abs/2603.14101}{{\tt
  arXiv:2603.14101}}.
\bibitem[{Budzinskiy et~al.(2025)Budzinskiy, Kazeev, and
  Olshanskii}]{budzinskiy2025cross}
\bibinfo{author}{S.~Budzinskiy}, \bibinfo{author}{V.~Kazeev},
  \bibinfo{author}{M.~Olshanskii}, \bibinfo{title}{Low-rank cross approximation
  of function-valued tensors for reduced-order modeling of parametric {PDE}s},
  \bibinfo{year}{2025}. \href{http://arxiv.org/abs/2511.22650}{{\tt
  arXiv:2511.22650}}.
\bibitem[{Choi and Carlberg(2019)}]{choi2019spacetime}
\bibinfo{author}{Y.~Choi}, \bibinfo{author}{K.~Carlberg},
\newblock \bibinfo{title}{Space--time least-squares {P}etrov--{G}alerkin
  projection for nonlinear model reduction},
\newblock \bibinfo{journal}{SIAM Journal on Scientific Computing}
  \bibinfo{volume}{41} (\bibinfo{year}{2019}) \bibinfo{pages}{A26--A58}.
  \DOIprefix\doi{10.1137/17M1120531}.
\bibitem[{Gopalan and Wikle(2022)}]{gopalan2022hosvd}
\bibinfo{author}{G.~Gopalan}, \bibinfo{author}{C.~K. Wikle},
\newblock \bibinfo{title}{A higher-order singular value decomposition tensor
  emulator for spatiotemporal simulators},
\newblock \bibinfo{journal}{Journal of Agricultural, Biological and
  Environmental Statistics} \bibinfo{volume}{27} (\bibinfo{year}{2022})
  \bibinfo{pages}{22--45}. \DOIprefix\doi{10.1007/s13253-021-00459-x}.
\bibitem[{Karas{\"o}zen et~al.(2022)Karas{\"o}zen, M{\"u}layim, and
  Uzunca}]{karasozen2022nonintrusive}
\bibinfo{author}{B.~Karas{\"o}zen}, \bibinfo{author}{G.~M{\"u}layim},
  \bibinfo{author}{M.~Uzunca},
\newblock \bibinfo{title}{Nonintrusive model order reduction for
  cross-diffusion systems},
\newblock \bibinfo{journal}{Communications in Nonlinear Science and Numerical
  Simulation} \bibinfo{volume}{115} (\bibinfo{year}{2022})
  \bibinfo{pages}{106734}. \DOIprefix\doi{10.1016/j.cnsns.2022.106734}.
\bibitem[{Barrag{\'a}n et~al.(2025)Barrag{\'a}n, Hetherington, Sengupta,
  Abad{\'i}a-Heredia, Garicano-Mena, and Le~Clainche}]{barragan2025motif}
\bibinfo{author}{G.~Barrag{\'a}n}, \bibinfo{author}{A.~Hetherington},
  \bibinfo{author}{A.~Sengupta}, \bibinfo{author}{R.~Abad{\'i}a-Heredia},
  \bibinfo{author}{J.~Garicano-Mena}, \bibinfo{author}{S.~Le~Clainche},
  \bibinfo{title}{{MoTIF}: A mode-structured tensor framework for
  multi-parametric approximation, super-resolution and forecasting of unsteady
  systems}, \bibinfo{year}{2025}. \href{http://arxiv.org/abs/2510.25625}{{\tt
  arXiv:2510.25625}}.
\bibitem[{Sasagawa and Tanaka(2023)}]{sasagawa2023tucker}
\bibinfo{author}{T.~Sasagawa}, \bibinfo{author}{M.~Tanaka},
\newblock \bibinfo{title}{Construction of a reduced-order model based on tensor
  decomposition and its application to airbag deployment simulations},
\newblock \bibinfo{journal}{Scientific Reports} \bibinfo{volume}{13}
  (\bibinfo{year}{2023}) \bibinfo{pages}{11241}.
  \DOIprefix\doi{10.1038/s41598-023-38393-2}.
\bibitem[{Shaa et~al.(2026)Shaa, Guet, Yang, Albergel, and
  Ribstein}]{shaa2026hybrid}
\bibinfo{author}{A.~Shaa}, \bibinfo{author}{C.~Guet},
  \bibinfo{author}{X.~Yang}, \bibinfo{author}{A.~Albergel},
  \bibinfo{author}{B.~Ribstein},
\newblock \bibinfo{title}{Hybrid neural interpolation of a sequence of wind
  flows},
\newblock \bibinfo{journal}{Engineering Applications of Computational Fluid
  Mechanics} \bibinfo{volume}{20} (\bibinfo{year}{2026})
  \bibinfo{pages}{2682648}. \DOIprefix\doi{10.1080/19942060.2026.2682648}.
\bibitem[{Eld{\'e}n and Savas(2009)}]{elden2009newton}
\bibinfo{author}{L.~Eld{\'e}n}, \bibinfo{author}{B.~Savas},
\newblock \bibinfo{title}{A {N}ewton--{G}rassmann method for computing the best
  multilinear rank-$(r_1,r_2,r_3)$ approximation of a tensor},
\newblock \bibinfo{journal}{SIAM Journal on Matrix Analysis and Applications}
  \bibinfo{volume}{31} (\bibinfo{year}{2009}) \bibinfo{pages}{248--271}.
  \DOIprefix\doi{10.1137/070688316}.
\bibitem[{Ishteva et~al.(2011)Ishteva, Absil, Van~Huffel, and
  De~Lathauwer}]{ishteva2011best}
\bibinfo{author}{M.~Ishteva}, \bibinfo{author}{P.-A. Absil},
  \bibinfo{author}{S.~Van~Huffel}, \bibinfo{author}{L.~De~Lathauwer},
\newblock \bibinfo{title}{Best low multilinear rank approximation of
  higher-order tensors, based on the {R}iemannian trust-region scheme},
\newblock \bibinfo{journal}{SIAM Journal on Matrix Analysis and Applications}
  \bibinfo{volume}{32} (\bibinfo{year}{2011}) \bibinfo{pages}{115--135}.
  \DOIprefix\doi{10.1137/090764827}.
\bibitem[{De~Sterck and Howse(2016)}]{desterck2016tucker}
\bibinfo{author}{H.~De~Sterck}, \bibinfo{author}{A.~J.~M. Howse},
\newblock \bibinfo{title}{Nonlinearly preconditioned optimization on
  {G}rassmann manifolds for computing approximate {T}ucker tensor
  decompositions},
\newblock \bibinfo{journal}{SIAM Journal on Scientific Computing}
  \bibinfo{volume}{38} (\bibinfo{year}{2016}) \bibinfo{pages}{A997--A1018}.
  \DOIprefix\doi{10.1137/15M1037288}.
\bibitem[{Lui(2012)}]{lui2012product}
\bibinfo{author}{Y.~M. Lui},
\newblock \bibinfo{title}{Human gesture recognition on product manifolds},
\newblock \bibinfo{journal}{Journal of Machine Learning Research}
  \bibinfo{volume}{13} (\bibinfo{year}{2012}) \bibinfo{pages}{3297--3321}.
\bibitem[{Friderikos et~al.(2020)Friderikos, Olive, Baranger, Sagris, and
  David}]{friderikos2020grassmann}
\bibinfo{author}{O.~Friderikos}, \bibinfo{author}{M.~Olive},
  \bibinfo{author}{E.~Baranger}, \bibinfo{author}{D.~Sagris},
  \bibinfo{author}{C.~N. David},
\newblock \bibinfo{title}{A space-time {POD} basis interpolation on {G}rassmann
  manifolds for parametric simulations of rigid-viscoplastic {FEM}},
\newblock \bibinfo{journal}{MATEC Web of Conferences} \bibinfo{volume}{318}
  (\bibinfo{year}{2020}) \bibinfo{pages}{01043}.
  \DOIprefix\doi{10.1051/matecconf/202031801043}.
\bibitem[{Friderikos et~al.(2021)Friderikos, Olive, Baranger, Sagris, and
  David}]{friderikos2021spacetime}
\bibinfo{author}{O.~Friderikos}, \bibinfo{author}{M.~Olive},
  \bibinfo{author}{E.~Baranger}, \bibinfo{author}{D.~Sagris},
  \bibinfo{author}{C.~David},
\newblock \bibinfo{title}{A non-intrusive space-time interpolation from compact
  {S}tiefel manifolds of parametrized rigid-viscoplastic {FEM} problems},
\newblock \bibinfo{journal}{Computational Mechanics} \bibinfo{volume}{68}
  (\bibinfo{year}{2021}) \bibinfo{pages}{861--883}.
  \DOIprefix\doi{10.1007/s00466-021-02050-0}.
\bibitem[{Zhou et~al.(2013)Zhou, Li, and Zhu}]{Zhou2013}
\bibinfo{author}{H.~Zhou}, \bibinfo{author}{L.~Li}, \bibinfo{author}{H.~Zhu},
\newblock \bibinfo{title}{Tensor regression with applications in neuroimaging
  data analysis},
\newblock \bibinfo{journal}{Journal of the American Statistical Association}
  \bibinfo{volume}{108} (\bibinfo{year}{2013}) \bibinfo{pages}{540--552}.
  \DOIprefix\doi{10.1080/01621459.2013.776499}.
\bibitem[{Li et~al.(2018)Li, Xu, Zhou, and Li}]{Li2018TuckerRegression}
\bibinfo{author}{X.~Li}, \bibinfo{author}{D.~Xu}, \bibinfo{author}{H.~Zhou},
  \bibinfo{author}{L.~Li},
\newblock \bibinfo{title}{Tucker tensor regression and neuroimaging analysis},
\newblock \bibinfo{journal}{Statistics in Biosciences} \bibinfo{volume}{10}
  (\bibinfo{year}{2018}) \bibinfo{pages}{520--545}.
  \DOIprefix\doi{10.1007/s12561-018-9215-6}.
\bibitem[{Fang et~al.(2019)Fang, Paynabar, and Gebraeel}]{fang2017image}
\bibinfo{author}{X.~Fang}, \bibinfo{author}{K.~Paynabar},
  \bibinfo{author}{N.~Gebraeel},
\newblock \bibinfo{title}{Image-based prognostics using penalized tensor
  regression},
\newblock \bibinfo{journal}{Technometrics} \bibinfo{volume}{61}
  (\bibinfo{year}{2019}) \bibinfo{pages}{369--384}.
  \DOIprefix\doi{10.1080/00401706.2018.1527727}.
\bibitem[{Chinesta et~al.(2011)Chinesta, Ladev{\`e}ze, and
  Cueto}]{chinesta2011pgd}
\bibinfo{author}{F.~Chinesta}, \bibinfo{author}{P.~Ladev{\`e}ze},
  \bibinfo{author}{E.~Cueto},
\newblock \bibinfo{title}{A short review on model order reduction based on
  proper generalized decomposition},
\newblock \bibinfo{journal}{Archives of Computational Methods in Engineering}
  \bibinfo{volume}{18} (\bibinfo{year}{2011}) \bibinfo{pages}{395--404}.
  \DOIprefix\doi{10.1007/s11831-011-9064-7}.
\bibitem[{Zhang et~al.(2022)Zhang, Lu, Tang, and Liu}]{zhang2022hidennpgd}
\bibinfo{author}{L.~Zhang}, \bibinfo{author}{Y.~Lu}, \bibinfo{author}{S.~Tang},
  \bibinfo{author}{W.~K. Liu},
\newblock \bibinfo{title}{{HiDeNN-PGD}: Reduced-order hierarchical deep
  learning neural networks},
\newblock \bibinfo{journal}{Computer Methods in Applied Mechanics and
  Engineering} \bibinfo{volume}{389} (\bibinfo{year}{2022})
  \bibinfo{pages}{114414}. \DOIprefix\doi{10.1016/j.cma.2021.114414}.
\bibitem[{Lu et~al.(2023)Lu, Li, Zhang, Park, Mojumder, Knapik, Sang, Tang,
  Apley, Wagner, and Liu}]{lu2023chidenn}
\bibinfo{author}{Y.~Lu}, \bibinfo{author}{H.~Li}, \bibinfo{author}{L.~Zhang},
  \bibinfo{author}{C.~Park}, \bibinfo{author}{S.~Mojumder},
  \bibinfo{author}{S.~Knapik}, \bibinfo{author}{Z.~Sang},
  \bibinfo{author}{S.~Tang}, \bibinfo{author}{D.~W. Apley},
  \bibinfo{author}{G.~J. Wagner}, \bibinfo{author}{W.~K. Liu},
\newblock \bibinfo{title}{Convolution hierarchical deep-learning neural
  networks ({C-HiDeNN}): Finite elements, isogeometric analysis, tensor
  decomposition, and beyond},
\newblock \bibinfo{journal}{Computational Mechanics} \bibinfo{volume}{72}
  (\bibinfo{year}{2023}) \bibinfo{pages}{333--362}.
  \DOIprefix\doi{10.1007/s00466-023-02336-5}.
\bibitem[{Park et~al.(2025)Park, Saha, Guo, Zhang, Xie, Bessa, Qian, Chen,
  Wagner, Cao, Hughes, and Liu}]{park2025inn}
\bibinfo{author}{C.~Park}, \bibinfo{author}{S.~Saha}, \bibinfo{author}{J.~Guo},
  \bibinfo{author}{H.~Zhang}, \bibinfo{author}{X.~Xie}, \bibinfo{author}{M.~A.
  Bessa}, \bibinfo{author}{D.~Qian}, \bibinfo{author}{W.~Chen},
  \bibinfo{author}{G.~J. Wagner}, \bibinfo{author}{J.~Cao},
  \bibinfo{author}{T.~J.~R. Hughes}, \bibinfo{author}{W.~K. Liu},
\newblock \bibinfo{title}{Unifying machine learning and interpolation theory
  via interpolating neural networks},
\newblock \bibinfo{journal}{Nature Communications} \bibinfo{volume}{16}
  (\bibinfo{year}{2025}) \bibinfo{pages}{8753}.
  \DOIprefix\doi{10.1038/s41467-025-63790-8}.
\bibitem[{Guo et~al.(2025{\natexlab{a}})Guo, Xie, Park, Zhang, Politis, Domel,
  and Liu}]{guo2025inntd}
\bibinfo{author}{J.~Guo}, \bibinfo{author}{X.~Xie}, \bibinfo{author}{C.~Park},
  \bibinfo{author}{H.~Zhang}, \bibinfo{author}{M.~J. Politis},
  \bibinfo{author}{G.~Domel}, \bibinfo{author}{W.~K. Liu},
\newblock \bibinfo{title}{Interpolating neural network-tensor decomposition
  ({INN-TD}): A scalable and interpretable approach for large-scale
  physics-based problems},
\newblock in: \bibinfo{booktitle}{Proceedings of the 42nd International
  Conference on Machine Learning}, volume \bibinfo{volume}{267} of
  \textit{\bibinfo{series}{Proceedings of Machine Learning Research}},
  \bibinfo{publisher}{PMLR}, \bibinfo{year}{2025}{\natexlab{a}}, pp.
  \bibinfo{pages}{21138--21162}.
\bibitem[{Guo et~al.(2025{\natexlab{b}})Guo, Domel, Park, Zhang, Gumus, Lu,
  Wagner, Qian, Cao, Hughes, and Liu}]{guo2025taps}
\bibinfo{author}{J.~Guo}, \bibinfo{author}{G.~Domel},
  \bibinfo{author}{C.~Park}, \bibinfo{author}{H.~Zhang}, \bibinfo{author}{O.~C.
  Gumus}, \bibinfo{author}{Y.~Lu}, \bibinfo{author}{G.~J. Wagner},
  \bibinfo{author}{D.~Qian}, \bibinfo{author}{J.~Cao},
  \bibinfo{author}{T.~J.~R. Hughes}, \bibinfo{author}{W.~K. Liu},
\newblock \bibinfo{title}{Tensor-decomposition-based a priori surrogate
  ({TAPS}) modeling for ultra large-scale simulations},
\newblock \bibinfo{journal}{Computer Methods in Applied Mechanics and
  Engineering} \bibinfo{volume}{444} (\bibinfo{year}{2025}{\natexlab{b}})
  \bibinfo{pages}{118101}. \DOIprefix\doi{10.1016/j.cma.2025.118101}.
\bibitem[{Koch and Lubich(2010)}]{koch2010dynamical}
\bibinfo{author}{O.~Koch}, \bibinfo{author}{C.~Lubich},
\newblock \bibinfo{title}{Dynamical tensor approximation},
\newblock \bibinfo{journal}{SIAM Journal on Matrix Analysis and Applications}
  \bibinfo{volume}{31} (\bibinfo{year}{2010}) \bibinfo{pages}{2360--2375}.
  \DOIprefix\doi{10.1137/09076578X}.
\bibitem[{Uschmajew and Vandereycken(2013)}]{uschmajew2013geometry}
\bibinfo{author}{A.~Uschmajew}, \bibinfo{author}{B.~Vandereycken},
\newblock \bibinfo{title}{The geometry of algorithms using hierarchical
  tensors},
\newblock \bibinfo{journal}{Linear Algebra and its Applications}
  \bibinfo{volume}{439} (\bibinfo{year}{2013}) \bibinfo{pages}{133--166}.
  \DOIprefix\doi{10.1016/j.laa.2013.03.016}.
\bibitem[{Goutaudier et~al.(2023)Goutaudier, Nobile, and
  Schiffmann}]{goutaudier2023orbi}
\bibinfo{author}{D.~Goutaudier}, \bibinfo{author}{F.~Nobile},
  \bibinfo{author}{J.~Schiffmann},
\newblock \bibinfo{title}{A new method to interpolate {POD} reduced bases:
  Application to the parametric model order reduction of a gas bearings
  supported rotor},
\newblock \bibinfo{journal}{International Journal for Numerical Methods in
  Engineering} \bibinfo{volume}{124} (\bibinfo{year}{2023})
  \bibinfo{pages}{4141--4170}. \DOIprefix\doi{10.1002/nme.7305}.
\bibitem[{Du and Zhang(2026)}]{du2026mrpwi}
\bibinfo{author}{L.~Du}, \bibinfo{author}{S.~Zhang},
  \bibinfo{title}{Mode-realigned pointwise interpolation ({MRPWI}) for
  efficient {POD}-{G}alerkin parametric reduced-order models},
  \bibinfo{year}{2026}. \href{http://arxiv.org/abs/2604.25955}{{\tt
  arXiv:2604.25955}}.
\bibitem[{Lyu et~al.(2026)Lyu, Yu, Kim, Zhou, Park, and Shi}]{lyu2026diffats}
\bibinfo{author}{J.~Lyu}, \bibinfo{author}{T.~Yu}, \bibinfo{author}{B.~Kim},
  \bibinfo{author}{L.~Zhou}, \bibinfo{author}{C.~Park},
  \bibinfo{author}{N.~Shi}, \bibinfo{title}{{DiffATS}: Diffusion in aligned
  tensor space}, \bibinfo{year}{2026}.
  \href{http://arxiv.org/abs/2605.09275}{{\tt arXiv:2605.09275}}.
\bibitem[{Karcher(1977)}]{karcher1977}
\bibinfo{author}{H.~Karcher},
\newblock \bibinfo{title}{Riemannian center of mass and mollifier smoothing},
\newblock \bibinfo{journal}{Communications on Pure and Applied Mathematics}
  \bibinfo{volume}{30} (\bibinfo{year}{1977}) \bibinfo{pages}{509--541}.
  \DOIprefix\doi{10.1002/cpa.3160300502}.
\bibitem[{Afsari(2011)}]{afsari2011}
\bibinfo{author}{B.~Afsari},
\newblock \bibinfo{title}{Riemannian {$L^p$} center of mass: Existence,
  uniqueness, and convexity},
\newblock \bibinfo{journal}{Proceedings of the American Mathematical Society}
  \bibinfo{volume}{139} (\bibinfo{year}{2011}) \bibinfo{pages}{655--673}.
  \DOIprefix\doi{10.1090/S0002-9939-2010-10541-5}.
\bibitem[{Gower and Dijksterhuis(2004)}]{gower2004procrustes}
\bibinfo{author}{J.~C. Gower}, \bibinfo{author}{G.~B. Dijksterhuis},
  \bibinfo{title}{Procrustes Problems}, volume~\bibinfo{volume}{30} of
  \textit{\bibinfo{series}{Oxford Statistical Science Series}},
  \bibinfo{publisher}{Oxford University Press}, \bibinfo{address}{Oxford},
  \bibinfo{year}{2004}.
  \DOIprefix\doi{10.1093/acprof:oso/9780198510581.001.0001}.
\bibitem[{Arlot and Celisse(2010)}]{arlot2010survey}
\bibinfo{author}{S.~Arlot}, \bibinfo{author}{A.~Celisse},
\newblock \bibinfo{title}{A survey of cross-validation procedures for model
  selection},
\newblock \bibinfo{journal}{Statistics Surveys} \bibinfo{volume}{4}
  (\bibinfo{year}{2010}) \bibinfo{pages}{40--79}.
  \DOIprefix\doi{10.1214/09-SS054}.
\bibitem[{Hackbusch(2012)}]{Hackbusch2012}
\bibinfo{author}{W.~Hackbusch}, \bibinfo{title}{Tensor Spaces and Numerical
  Tensor Calculus}, volume~\bibinfo{volume}{42} of
  \textit{\bibinfo{series}{Springer Series in Computational Mathematics}},
  \bibinfo{publisher}{Springer}, \bibinfo{address}{Berlin, Heidelberg},
  \bibinfo{year}{2012}. \DOIprefix\doi{10.1007/978-3-642-28027-6}.
\bibitem[{Luo et~al.(2023{\natexlab{a}})Luo, Chen, and Zhang}]{luo2023cfdbench}
\bibinfo{author}{Y.~Luo}, \bibinfo{author}{Y.~Chen},
  \bibinfo{author}{Z.~Zhang}, \bibinfo{title}{{CFDB}ench: A large-scale
  benchmark for machine learning methods in fluid dynamics},
  \bibinfo{year}{2023}{\natexlab{a}}.
  \href{http://arxiv.org/abs/2310.05963}{{\tt arXiv:2310.05963}}.
\bibitem[{Luo et~al.(2023{\natexlab{b}})Luo, Chen, and
  Zhang}]{cfdbench2023data}
\bibinfo{author}{Y.~Luo}, \bibinfo{author}{Y.~Chen},
  \bibinfo{author}{Z.~Zhang}, \bibinfo{title}{{CFDB}ench dataset},
  \bibinfo{howpublished}{Hugging Face Datasets,
  \url{https://huggingface.co/datasets/chen-yingfa/CFDBench}},
  \bibinfo{year}{2023}{\natexlab{b}}. \bibinfo{note}{Apache-2.0 license,
  revision 09d3f237, accessed 24 September 2026}.
\bibitem[{Eckart and Young(1936)}]{Eckart1936}
\bibinfo{author}{C.~Eckart}, \bibinfo{author}{G.~Young},
\newblock \bibinfo{title}{The approximation of one matrix by another of lower
  rank},
\newblock \bibinfo{journal}{Psychometrika} \bibinfo{volume}{1}
  (\bibinfo{year}{1936}) \bibinfo{pages}{211--218}.
  \DOIprefix\doi{10.1007/BF02288367}.
\bibitem[{Halko et~al.(2011)Halko, Martinsson, and Tropp}]{halko2011randomized}
\bibinfo{author}{N.~Halko}, \bibinfo{author}{P.-G. Martinsson},
  \bibinfo{author}{J.~A. Tropp},
\newblock \bibinfo{title}{Finding structure with randomness: Probabilistic
  algorithms for constructing approximate matrix decompositions},
\newblock \bibinfo{journal}{SIAM Review} \bibinfo{volume}{53}
  (\bibinfo{year}{2011}) \bibinfo{pages}{217--288}.
  \DOIprefix\doi{10.1137/090771806}.
\bibitem[{Reisi~Gahrooei et~al.(2021)Reisi~Gahrooei, Yan, Paynabar, and
  Shi}]{reisi2021mtot}
\bibinfo{author}{M.~Reisi~Gahrooei}, \bibinfo{author}{H.~Yan},
  \bibinfo{author}{K.~Paynabar}, \bibinfo{author}{J.~Shi},
\newblock \bibinfo{title}{Multiple tensor-on-tensor regression: An approach for
  modeling processes with heterogeneous sources of data},
\newblock \bibinfo{journal}{Technometrics} \bibinfo{volume}{63}
  (\bibinfo{year}{2021}) \bibinfo{pages}{147--159}.
  \DOIprefix\doi{10.1080/00401706.2019.1708463}.

\end{thebibliography}
\end{document}